\documentclass{article}
\usepackage{graphicx} 
\usepackage{amsfonts,amsmath,amsthm}
\usepackage{hyperref}
\usepackage{comment}

\newtheorem{theorem}{Theorem}[section]

\newtheorem{proposition}[theorem]{Proposition}

\usepackage{xcolor}
\usepackage{multirow}
\usepackage{subcaption}

\title{On-Line Policy Iteration with Trajectory-Driven Policy Generation\footnote{
Y. Li and D. Bertsekas are with the Fulton School of Engineering, Arizona State University, Tempe, AZ. F. Chen, Y. Li, and C. Fan are with the Department of Aeronautics and Astronautics and Laboratory for Information and Decision Systems, MIT, Cambridge, MA.}}
\author{Yuchao Li, Fei Chen, Yingke Li, Chuchu Fan, and Dimitri Bertsekas}
\date{November 2026}

\begin{document}
\maketitle
\begin{abstract}
We consider deterministic finite-horizon optimal control problems with a fixed initial state. We introduce an on-line policy iteration method, which, starting from a given policy, however obtained, generates a sequence of cost-improving policies and corresponding trajectories. Each policy produces a trajectory, which is used in turn to generate data for training the  next policy. The method is motivated by problems that are repeatedly solved starting from the same initial state, including discrete optimization and path planning for repetitive tasks. For such problems, the method is fast enough to be used on-line. Under a natural consistency condition, we show that the sequence of costs of the generated policies is monotonically improving for the given initial state (but not necessarily for other states). We illustrate our results with computational studies from combinatorial optimization and  3-dimensional path planning for drones and a robot arm in the presence of obstacles. We also discuss briefly a stochastic counterpart of our algorithm. Our proposed framework combines elements of rollout and policy iteration with flexible trajectory-based policy representations, and applies to problems involving a single as well as multiple decision makers. It also provides a principled way to train neural network-based policies using trajectory data, while preserving monotonic cost improvement.
\end{abstract}

\section{Introduction}\label{sec:intro}
A general deterministic finite-horizon optimal control problem involves the discrete-time dynamic system
\begin{equation}
\label{eq:finite_dynamics}
    x_{k+1}=f_k(x_k,u_k),\qquad k=0,1,\dots,N-1,
\end{equation}
where the state $x_k$ at the $k$th stage is an element of some state space, and the control $u_k$ is an element of some control space. The control $u_k$ is constrained to take values in a given subset $U_k(x_k)$, which is nonempty and depends on the current state $x_k$. The cost of the $k$th stage is denoted by $g_k(x_k,u_k)$. Upon reaching state $x_N$ at the $N$th stage, there is a terminal cost $g_N(x_N)$.

We consider policies of the form 
$$\pi=\{\mu_0,\mu_1,\dots,\mu_{N-1}\},$$
where $\mu_k$ maps states $x_k$ into controls $u_k=\mu_k(x_k)$, and satisfies the control constraints $\mu_k(x_k)\in U_k(x_k)$ for all $x_k$. Given an initial state $x_0$ and a policy $\pi=\{\mu_0,\dots,\mu_{N-1}\}$, the cost of $\pi$ starting from $x_0$ is given by
$$J_\pi(x_0)=g_N(x_N)+\sum_{k=0}^{N-1}g_k\big(x_k,\mu_k(x_k)\big),$$
where $x_{k+1}=f_k\big(x_k,\mu_k(x_k)\big)$, $k=1,\dots,N-1$; cf. Eq.~\eqref{eq:finite_dynamics}. The optimal cost starting from $x_0$ is denoted by $J^*(x_0)$:
$$J^*(x_0)=\min_{\pi\in \Pi}J_\pi(x_0),$$
where $\Pi$ is the set of all policies. An optimal policy $\pi^*$ is one that attains the optimal cost for every $x_0$; i.e., 
$$J_{\pi^*}(x_0)=\min_{\pi\in \Pi}J_\pi(x_0),\qquad \hbox{for all }x_0.$$

In principle, the optimal cost function $J^*$ and optimal policy $\pi^*$ can be computed via dynamic programming (DP): starting with $J^*_N(x_N)=g_N(x_N)$ for all $x_N$, and for $k=N-1,\dots,0$, we set
\begin{equation}
    \label{eq:bellman}
    J^*_k(x_k)=\min_{u_k\in U_k(x_k)}\Big[g_k(x_k,u_k)+J^*_{k+1}\big(f_k(x_k,u_k)\big)\Big],\qquad \hbox{for all }x_k.
\end{equation}
The optimal cost function $J^*$ is equal to $J^*_0$, the final function generated by DP. Moreover, an optimal policy $\pi^*=\{\mu^*_0,\dots,\mu^*_{N-1}\}$ can be obtained via\footnote{Throughout the paper, wherever minima appear in various relations, we assume that they are attained. Therefore, we use the notation $\min$ in place of $\inf$. Note that the minimizers need not be unique.}
\begin{equation}
    \label{eq:optimal_mu}
    \mu^*_k(x_k)\in \arg\min_{u_k\in U_k(x_k)}\Big[g_k(x_k,u_k)+J^*_{k+1}\big(f_k(x_k,u_k)\big)\Big],
\end{equation}
for all $x_k$, $k=0,1,\dots,N-1$. However, the DP algorithm is computationally prohibitive for most practical problems.

An alternative to the DP algorithm is policy iteration (PI). Starting with a policy $\pi^0=\{\mu_0^0,\mu_1^0,\dots,\mu_{N-1}^0\}$, the PI algorithm generates a sequence of policies $\pi^0,\pi^1,\dots$, which are \emph{improving} in the sense that
$$J_{\pi^{k+1}}(x_0)\leq J_{\pi^{k}}(x_0),\qquad \text{for all $x_0$ and $k$},$$
while it can be shown that $J_{\pi^N}=J^*$ (i.e., the policy obtained after $N$ policy iterations is optimal); cf. Prop.~\ref{prop:pi_property}. However, the PI algorithm requires more computation than the DP algorithm, so it makes little sense to use it for solving the problem exactly. Still, the PI algorithm has served as the starting point for  approximate solution approaches that have been popular in reinforcement learning (RL). 

There have been recent proposals of approximate PI schemes that compute a sequence of policies \emph{on-line, during the system's operation}. For example,  methods of this type have been proposed in \cite{rosolia2017learning}, \cite{li2021data}, \cite{bertsekas2021line}, \cite{rosolia2022optimality}, and \cite{bao2024value}. The purpose of the present paper is to introduce a  form of on-line PI, which is related to  these works, but contains some substantially new ideas relating to cost improvement guarantees. In particular, our algorithm computes an improved policy for states that are encountered during operation of the system, and applies to problems with a fixed initial state. The requirement of a fixed initial state may appear restrictive at first sight, but it arises in many important contexts, such as discrete optimization and control for repetitive tasks. Our algorithm assumes that a policy can be obtained by using \emph{a single complete trajectory $\{x_0,u_0,x_1,\dots,u_{N-1},x_N\}$}. In particular, we will propose a method that uses sampling around this trajectory, and then trains a policy network to learn a new and improved policy and associated trajectory. We will demonstrate our method with challenging applications involving discrete optimization and drone robotics, including path planning problems that involve multiple drones, which have not been addressed in the literature.

The remainder of the paper is organized as follows. In Section~\ref{sec:background}, we provide a brief review of exact and approximate PI algorithms, and place our contributions in context. In Section~\ref{sec:online_pi}, we introduce our on-line PI algorithm and its variants, and analyze their performance guarantees. In Section~\ref{sec:compute}, we present computational results with our algorithm applied to a discrete optimization problem and path-planning problems involving a single as well as multiple decision makers.

\section{Exact and Approximate Policy Iteration}\label{sec:background}
In this section, we review exact and approximate PI algorithms that are related to our proposed schemes. In addition, we describe the theoretical properties of the exact PI algorithm for finite-horizon problems, laying the foundation for our on-line approximate PI algorithm and its variants. 

\subsection{Related Work}
The PI algorithm dates back to Bellman and his classical book \cite{bellman1957dynamic}, where it was referred to as `approximation in policy space.' When applied to problems with a finite horizon, the PI algorithm requires much more computation than the exact DP algorithm. Therefore, it makes little sense to use PI when the DP algorithm is tractable. Still, PI has been the starting point for designing various approximate solution methods. Broadly speaking, we divide these algorithms into two categories: \emph{off-line algorithms}, which are executed before the start of the system's operation (and hence are not subject to computation time restrictions), and \emph{on-line algorithms}, which are fast enough to be used while the system is in operation.

The off-line approximate PI algorithms typically approximate the PI algorithm by using some form of cost function approximation. The purpose of the approximation is to deal with a large state space. Instead of computing policies, these algorithms often collect data for some sample states by using approximate forms of PI and train the function approximators afterwards. Successful applications of these methods include early works such as the TD-Gammon program by Tesauro \cite{tesauro1995temporal}, and later by Tesauro and Galperin \cite{tesauro1996line}, the application to Tetris by Bertsekas and Ioffe \cite{bertsekas1996temporal}, and later by Scherrer et al. \cite{scherrer2015approximate}, as well as the more recent AlphaGo and AlphaZero programs \cite{silver2016mastering,silver2017mastering}.

The off-line algorithms involve multiple iterations of approximate PI. By contrast,  on-line algorithms start with a base policy and execute only one iteration of PI, either exactly or approximately, and only for those states encountered during system operation. In fact, after completion of the off-line computation, both TD-Gammon \cite{tesauro1996line} and AlphaGo \cite{silver2016mastering} perform on-line a policy improvement computation to obtain a much better policy. The on-line exact and approximate execution of an iteration of PI is often referred to as \emph{rollout}, and has been analyzed and applied in a variety of contexts; see \cite{bertsekas1997rollout,bertsekas1999rollout,secomandi2001rollout,bertsekas2021multiagent,bai2023rollout,hammar2025adaptive}. Moreover, some important forms of model predictive control \cite{sznaier1987suboptimal,keerthi1988optimal,chen1998quasi}, a prominent control system design methodology, can also be viewed as conducting a single iteration of PI; see \cite{bertsekas2005dynamic,bertsekas2022lessons}. Still, one iteration of approximate PI may not provide sufficient improvement, depending on the application at hand.

To enable on-line multiple policy improvements, there have been proposals of PI algorithms that can be applied repeatedly during the system operation. Examples include \cite{rosolia2017learning} and \cite{bertsekas2021line}. Focusing on problems with a fixed initial state, Rosolia and Borrelli \cite{rosolia2017learning} proposed a form of approximate PI that uses the trajectory of the previous iteration and improves upon it. For general stochastic optimal control problems with an infinite horizon, the last author introduced an on-line PI algorithm \cite{bertsekas2021line} for discounted infinite horizon Markovian decision problems, which executes policy improvement only for states encountered during real-time operation. The algorithm  was shown to be convergent to the optimal policy under reasonable assumptions. The present work is inspired by this research but differs in three important ways:
\begin{itemize}
    \item[1)] By introducing a mechanism that generates policies based on complete trajectory data, our method allows more flexible representations of policies; 
    \item[2)] By providing a natural and checkable condition for preserving the cost-improvement guarantee, our method helps to guide the design and training of neural-network-based policies;
    \item[3)] By providing a computational procedure for dealing with continuous control spaces, our method applies to highly nonlinear problems and scales well with control dimension.  
\end{itemize}

Similar to the scheme proposed in \cite{bertsekas2021line}, the on-line PI method can be used even when the state equation $f_k$ and/or the stage and terminal cost functions $g_k$ are available only as `black boxes.' In particular, our method does not require analytical expressions for $f_k$ and $g_k$, provided that state trajectories and the corresponding stage and terminal costs can be simulated efficiently in real time. We illustrate this point in Section~\ref{sec:4_4}.

\subsection{Policy Iteration and Cost Improvement}
We briefly review the PI algorithm and its cost improvement property. For a textbook treatment of the subject, see \cite{bertsekas2020rollout,bertsekas2025course}. As discussed in Section~\ref{sec:intro}, the PI algorithm starts with a policy $\pi^0$, and generates a sequence of policies $\pi^0,\pi^1,\dots$. At a typical iteration, given policy $\pi^\ell=\{\mu_0^{\ell},\dots,\mu_{N-1}^{\ell}\}$, we first perform the \emph{policy evaluation} step, where we set 
\begin{equation}
    \label{eq:pe_n}
    J_{N,\pi^\ell}(x_N)=g_N(x_N),\qquad \hbox{for all }x_N,
\end{equation}
and
\begin{equation}
    \label{eq:pe_k}
    J_{k,\pi^\ell}(x_k)=g_N(x_N)+\sum_{i=k}^{N-1}g_i\big(x_i,\mu_i^\ell(x_i)\big),\;\hbox{for all }x_k,\;k=0,\dots,N-1,
\end{equation}
with $x_{i+1}=f_i\big(x_i,\mu_i^\ell(x_i)\big)$, $i=k,\dots,N-1$. We then carry out the \emph{policy improvement} step, where we set $\pi^{\ell+1}=\{\mu_0^{\ell+1},\dots,\mu_{N-1}^{\ell+1}\}$ with
\begin{equation}
    \label{eq:pi}
    \mu_{k}^{\ell+1}(x_k)\in \arg\min_{u_k\in U_k(x_k)}\Big[g_k(x_k,u_k)+J_{k+1,\pi^\ell}\big(f_k(x_k,u_k)\big)\Big],\qquad \hbox{for all }x_k.
\end{equation}

The following classical result holds for the PI algorithm.\footnote{For finite-horizon problems, the PI results of this section and the analysis developed in the subsequent sections also apply in the presence of a discount factor $\alpha\in(0,1)$. This follows by replacing $g_k$ and $g_N$ with $\alpha^k g_k$ and $\alpha^N g_N$, respectively.}
\begin{proposition}[Cost Improvement and Convergence of PI]\label{prop:pi_property}
    Consider the sequence of policies $\pi^0,\pi^1,\dots$, generated by the PI algorithm \eqref{eq:pe_n}-\eqref{eq:pi}. The policies are improving in the sense that 
    \begin{equation}
        \label{eq:pi_imp}
        J_{k,\pi^{\ell+1}}(x_k)\leq J_{k,\pi^{\ell}}(x_k), \qquad \hbox{for all $x_k$, $k$, and $\ell$}.
    \end{equation}
    Moreover, PI converges in $N$ iterations, i.e., $J_{\pi^N}=J^*$, and the policy $\pi^N$ is optimal.
\end{proposition}

\begin{proof}
    The cost improvement property of PI for infinite horizon problems is well-known. The cost improvement and finite termination of PI for finite horizon problems is also known but has not been discussed broadly. We provide here a short proof for completeness.

    We show finite termination by induction. By definition, we have $J_{N,\pi^0}=J_N^*$, where $J_N^*$ is the function defined in the DP algorithm; see Section~\ref{sec:intro}. As a result, $\mu_{N-1}^1$ attains the minimum in Eq.~\eqref{eq:bellman} for $k=N-1$. In addition, $J_{k,\pi^1}=J_k^*$, $k=N,N-1$. Suppose that $J_{k,\pi^\ell}=J_k^*$, $k=N,\dots,N-\ell$. By comparing Eqs.~\eqref{eq:optimal_mu} and \eqref{eq:pi}, we have that $\mu_k^{\ell+1}$ satisfies Eq.~\eqref{eq:optimal_mu} for $k=N-1,\dots,N-\ell-1$. Starting from $k=N$ and using backward induction, we obtain $J_{k,\pi^{\ell+1}}=J_k^*$, $k=N,\dots,N-\ell-1$. This concludes the proof.
\end{proof}

As discussed in Section~\ref{sec:intro}, the PI algorithm \eqref{eq:pe_n}-\eqref{eq:pi}  requires policy evaluation and policy improvement computations for all states. This makes it intractable for problems with a large state space. In what follows, we will describe our on-line PI algorithm and some variants, where the aforementioned computation is carried out only for states encountered during real-time operation.

\section{On-Line Policy Iteration}\label{sec:online_pi}
In this section, we describe our on-line PI algorithm. We introduce some mild conditions under which the policy improvement property of exact PI is preserved.

\subsection{On-Line Policy Iteration and Cost Improvement}
We will assume  throughout the remainder of the paper that \emph{the initial state $x_0$ is fixed and that the optimal cost $J^*(x_0)$ is finite}. 
To describe our algorithm, we introduce a few definitions. We say that a trajectory $\{x_0,u_0,x_1,\dots,u_{N-1},x_N\}$ is \emph{feasible} if it satisfies
$$x_{k+1}=f_k(x_k,u_k),\quad u_k\in U_k(x_k),\qquad k=0,\dots,N-1.$$
Our algorithm makes use of a special function, called a \emph{generator} and denoted by $\mathcal{G}$. It takes as an input a feasible trajectory $\{x_0,u_0,x_1,\dots,u_{N-1},x_N\}$ and generates a policy $\pi$, i.e., 
\begin{equation}
    \label{eq:generator_def}
    \pi=\mathcal{G}(x_0,u_0,x_1,\dots,u_{N-1},x_N).
\end{equation}

Our on-line PI algorithm starts with a given policy $\pi^0$. At the $\ell$th iteration, we have a policy $\pi^\ell$. We apply this policy to obtain a corresponding feasible trajectory $\{x_0,u_0^{\ell+1},x_1^{\ell+1},\dots,u_{N-1}^{\ell+1},x_N^{\ell+1}\}$ and then use the generator to obtain a new policy $\pi^{\ell+1}$ through the following steps:
\begin{itemize}
    \item[(1)] Starting with the fixed initial state $x_0$, we compute $u_0^{\ell+1}$ via the minimization
    \begin{equation}
        \label{eq:pi_im_t1}
        u_0^{\ell+1}\in \arg\min_{u_0\in U_0(x_0)}\Big[g_0(x_0,u_0)+J_{1,\pi^\ell}\big(f_0(x_0,u_0)\big)\Big],
    \end{equation}
    with preference given to $\mu_0^\ell(x_0)$ in case of a tie, and we set 
    \begin{equation}
        \label{eq:sl_x1}
        x_1^{\ell+1}=f_0(x_0,u_0^{\ell+1}).
    \end{equation}
    \item[(2)] Going forward, for $k=1,\dots,N-1$, we compute $u_k^{\ell+1}$ via the minimization
    \begin{equation}
        \label{eq:pi_im_tk}
        u_k^{\ell+1}\in \arg\min_{u_k\in U_k(x_k^{\ell+1})}\Big[g_k(x_k^{\ell+1},u_k)+J_{k+1,\pi^{\ell}}\big(f_k(x_k^{\ell+1},u_k)\big)\Big],
    \end{equation}
    with preference given to $\mu_k^\ell(x_k^{\ell+1})$ in case of a tie, and we set
    \begin{equation}
        \label{eq:sl_xk}
        x_{k+1}^{\ell+1}=f_k(x_k^{\ell+1},u_k^{\ell+1}).
    \end{equation}
    \item[(3)] Finally, we obtain the new policy $\pi^{\ell+1}$ by using the generator $\mathcal G$, i.e.,  
    \begin{equation}
        \label{eq:sl_g}
        \pi^{\ell+1}=\mathcal{G}(x_0,u_0^{\ell+1},x_1^{\ell+1},\dots,u_{N-1}^{\ell+1},x_N^{\ell+1}).
    \end{equation}
\end{itemize}

Note that since the computation involves states on a single trajectory and moves forward in stages, it is well-suited for path planning or control problems for repetitive tasks. In such contexts, starting with the fixed initial state, the algorithm is carried out to complete the task, and the complete trajectory is obtained as a by-product. Provided that the computation of a new policy through the use of the generator is relatively fast, the on-line PI algorithm can be carried out with little additional computation. Regarding the values of cost functions involved in the calculations \eqref{eq:pi_im_t1} and \eqref{eq:pi_im_tk}, they can be obtained via either simple calculation, as in discrete optimization applications, or through real-time simulation, as in path planning problems.

The important theoretical question is whether the policies obtained through the generator are cost improving. This is guaranteed under a mild condition. In particular, we say that the generator is \emph{consistent} if given any feasible trajectory $\{x_0,u_0,x_1,\dots,u_{N-1},x_N\}$, the generated policy $\pi=\{\mu_0,\mu_1,\dots,\mu_N\}$ satisfies
\begin{equation}
    \label{eq:consistency}
    \mu_k(x_k)=u_k,\qquad k=0,\dots,N-1.
\end{equation}
Thus, a generator $\mathcal G$ is consistent if \emph{the generated policy reproduces the tail of the given trajectory from every state along that trajectory}. We note that consistency is a structural requirement on the generator and is not automatically satisfied by an arbitrary policy parametrization. In Section~\ref{sec:4_2}, we construct a neural-network-based generator that satisfies the consistency condition exactly. A related tail-preservation idea was introduced in \cite{bertsekas1997rollout}. In the present work, consistency serves a different purpose: it ensures that the improvement of the computed trajectory via Eqs.~\eqref{eq:pi_im_t1}-\eqref{eq:sl_xk} is inherited by the generated policy. Figure~\ref{fig:consistency} illustrates how this improvement property may fail when the generator is inconsistent.

\begin{figure}[t]
    \centering
    \includegraphics[width=0.6\linewidth]{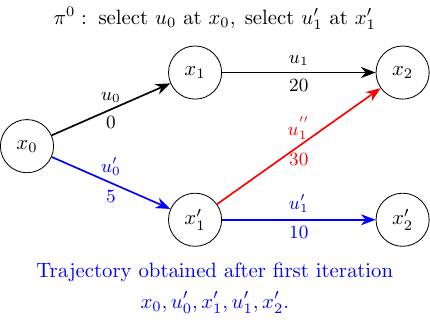}
    \caption{Example showing that the cost-improvement property of on-line PI may fail when the generator $\mathcal G$ is inconsistent. The example involves five states, represented by the nodes in the graph. The state transitions are given by the arcs, with the corresponding controls shown above the arcs and the stage costs shown below them. Let $\pi^0$ be the initial policy that selects $u_0$ at $x_0$ and $u_1'$ at $x_1'$ (so the trajectory under $\pi^0$ is shown in black). Then the trajectory produced by the first on-line PI iteration is $\{x_0, u_0', x_1', u_1', x_2'\}$ (shown in blue), which is optimal. Its total cost is $15$, smaller than the cost of $\pi^0$ at $x_0$, namely $J_{\pi^0}(x_0)=20$. Suppose now that the generator $\mathcal G$ takes this trajectory and produces a policy $\pi^1$ that selects $u_0'$ at $x_0$ and $u_1^{''}$ at $x_1'$. Then, starting from the trajectory state $x_1'$, the policy $\pi^1$ chooses a control (shown in red) different from the one in the given trajectory. As a result, $\pi^1$ does not preserve the optimality of the trajectory $\{x_0, u_0', x_1', u_1', x_2'\}$. It can be seen that $\pi^1$ is not consistent, since at state $x_1'$ it selects the control $u_1''$, which is not equal to $\mu_1(x_1')$, i.e., $u_1'$.}
    \label{fig:consistency}
\end{figure}

\begin{proposition}[Properties of On-Line PI]\label{prop:sl_property}
    Consider the sequence of policies $\pi^0,\pi^1,\dots$, computed via the on-line PI algorithm \eqref{eq:pi_im_t1}-\eqref{eq:sl_g}. Suppose that the generator is consistent. Then: 
    \begin{itemize}
        \item[(a)] The policies are improving at $x_0$ in the sense that 
    \begin{equation}
        \label{eq:sl_imp}
        J_{\pi^{\ell+1}}(x_0)\leq J_{\pi^{\ell}}(x_0), \qquad \hbox{for all $\ell$}, 
    \end{equation}
    and the sequence $\{J_{\pi^\ell}(x_0)\}$ is convergent.
    \item[(b)] If equality in \eqref{eq:sl_imp} holds for some $\bar \ell$, we have $\pi^{\ell+1}=\pi^\ell$ for all $\ell\geq \bar \ell\geq1$.
    \end{itemize}
    
\end{proposition}
\begin{proof}
    (a) Since the generator is consistent, we have for all $\ell$, 
    \begin{equation}
        \label{eq:pe_data}
        \begin{aligned}
            J_{0,\pi^{\ell}}(x_0)= & g_0\big(x_0,\mu_0^\ell(x_0)\big)+J_{1,\pi^\ell}(x_1^\ell),\\
            J_{k,\pi^{\ell}}(x_k^\ell)= & g_k\big(x_k^\ell,\mu_k^\ell(x_k^\ell)\big)+J_{k+1,\pi^\ell}(x_{k+1}^\ell),\quad k=1,\dots,N-1. 
        \end{aligned}
    \end{equation}
    We will prove by induction that
    \begin{equation}
        \label{eq:prop1_ineq}
        \begin{aligned}
            J_{0,\pi^{\ell+1}}(x_0)\leq& J_{0,\pi^{\ell}}(x_0),\\
            J_{k,\pi^{\ell+1}}(x_k^{\ell+1})\leq & J_{k,\pi^{\ell}}(x_k^{\ell+1}),\quad k=1,\dots,N.
        \end{aligned}
    \end{equation}
    Clearly, the inequality \eqref{eq:prop1_ineq} holds for $k=N$ since 
    $$J_{N,\pi^{\ell+1}}(x_N^{\ell+1})=g_N(x_N^{\ell+1})=J_{N,\pi^{\ell}}(x_N^{\ell+1}).$$
    Suppose that the inequality \eqref{eq:prop1_ineq} holds for $k+1$, where $k=1,2,\dots,N-1$. Then we have
    \begin{equation}
        \label{eq:proof_ineq}
        \begin{aligned}
        J_{k,\pi^{\ell+1}}(x_k^{\ell+1})=&g_k\big(x_k^{\ell+1},\mu_k^{\ell+1}(x_k^{\ell+1})\big)+J_{k+1,\pi^{\ell+1}}(x_{k+1}^{\ell+1})\\
        \leq & g_k\big(x_k^{\ell+1},\mu_k^\ell(x_k^{\ell+1})\big)+J_{k+1,\pi^{\ell}}(x_{k+1}^{\ell+1})\\
        =&g_k\big(x_k^{\ell+1},\mu_k^{\ell+1}(x_k^{\ell+1})\big)+J_{k+1,\pi^{\ell}}\Big(f_k\big(x_k^{\ell+1},\mu_k^{\ell+1}(x_k^{\ell+1})\big)\Big)\\
        =&\min_{u_k\in U_k(x_k^{\ell+1})}\Big[g_k(x_k^{\ell+1},u_k)+J_{k+1,\pi^{\ell}}\big(f_k(x_k^{\ell+1},u_k)\big)\Big]\\
        \leq&g_k\big(x_k^{\ell+1},\mu_k^{\ell}(x_k^{\ell+1})\big)+J_{k+1,\pi^{\ell}}\Big(f_k\big(x_k^{\ell+1},\mu_k^{\ell}(x_k^{\ell+1})\big)\Big)\\
        =&J_{k,\pi^{\ell}}(x_k^{\ell+1}),
    \end{aligned}
    \end{equation}
    where the first equality is due to Eq.~\eqref{eq:pe_data}, the first inequality holds by the induction hypothesis; the second and the third equalities hold by the consistency of the generator and the definition of $u_k^{\ell+1}$; the last equality is the DP equation for policy $\pi^\ell$. 

    Finally, applying inequality \eqref{eq:proof_ineq} with $k=0$ and $x_0$ in place of $x_k^{\ell+1}$, we have 
    $$J_{0,\pi^{\ell+1}}(x_0)\leq J_{0,\pi^{\ell}}(x_0).$$
    Since $J_{0,\pi^{\ell+1}}(x_0)=J_{\pi^{\ell+1}}(x_0)$ and $J_{0,\pi^{\ell}}(x_0)=J_{\pi^{\ell}}(x_0)$, we obtain the desired inequality \eqref{eq:sl_imp}. Since the sequence $\{J_{\pi^\ell}(x_0)\}$ is monotonically nonincreasing and lower bounded by $J^*(x_0)$, it must be convergent.

    (b) Suppose that equality holds in \eqref{eq:sl_imp} for some $\bar\ell\geq1$, i.e., 
    \begin{equation}
        \label{eq:equiv_con}
        J_{0,\pi^{\bar\ell+1}}(x_0)=J_{0,\pi^{\bar\ell}}(x_0).
    \end{equation}
    We will show by induction that for $k=1,\dots,N$,
    \begin{equation}
        \label{eq:converge}
            u_{k-1}^{\bar\ell+1}= u_{k-1}^{\bar\ell},\qquad 
            x_k^{\bar\ell+1}= x_k^{\bar\ell},\qquad
            J_{k,\pi^{\bar\ell+1}}(x_k^{\bar\ell})=  J_{k,\pi^{\bar\ell}}(x_k^{\bar\ell}).
    \end{equation}
    
    Since the inequality \eqref{eq:proof_ineq} holds with $k=0$, $\bar\ell$ in place of $\ell$, and $x_0$ in place of $x_k^{\ell+1}$, Eq.~\eqref{eq:equiv_con} implies that all the inequalities in \eqref{eq:proof_ineq} become equalities, i.e.,
    \begin{equation}
    \label{eq:induction_1}
         g_0\big(x_0,\mu_0^{\bar\ell}(x_0)\big)+J_{1,\pi^{\bar\ell+1}}(x_{1}^{\bar\ell+1})
        =  g_0\big(x_0,\mu_0^{\bar\ell}(x_0)\big)+J_{1,\pi^{\bar\ell}}(x_{1}^{\bar\ell+1}),
    \end{equation}
    \begin{equation}
    \label{eq:induction_2}
        \begin{aligned}
        \min_{u_k\in U_0(x_0)}\Big[g_0(x_0,u_0)&+J_{1,\pi^{\bar\ell}}\big(f_0(x_0,u_0)\big)\Big]\\
        =&g_0\big(x_0,\mu_0^{\bar\ell}(x_0\big)+J_{1,\pi^{\bar\ell}}\Big(f_0\big(x_0,\mu_0^{\bar\ell}(x_0)\big)\Big).
        \end{aligned}
    \end{equation}
    Equation~\eqref{eq:induction_1} implies that
    $$J_{1,\pi^{\bar\ell+1}}(x_{1}^{\bar\ell+1})= J_{1,\pi^{\bar\ell}}(x_{1}^{\bar\ell+1}),$$
    so Eq.~\eqref{eq:induction_2} implies $u_0^{\bar\ell+1}=\mu_0^{\bar\ell}(x_0)$ in view of our tie-breaking rule. Due to the consistency of the generator $\mathcal G$, we have $\mu_0^{\bar\ell}(x_0)=u_0^{\bar\ell}$, which yields $u_0^{\bar\ell+1}=u_0^{\bar\ell}$ and $x_1^{\bar\ell+1}=x_1^{\bar\ell}$. Thus, we have shown Eq.~\eqref{eq:converge} for $k=1$. 
    
    Suppose Eq.~\eqref{eq:converge} holds for the typical index $k$. Then the induction step can be shown by using similar arguments as above with $k$ in place of $0$, $x_k^{\bar\ell}$ in place of $x_0$, starting with the equality $J_{k,\pi^{\bar\ell+1}}(x_k^{\bar\ell+1})=  J_{k,\pi^{\bar\ell}}(x_k^{\bar\ell+1})$ and using also the fact that $x_k^{\bar\ell+1}=x_k^{\bar\ell}$. The details are left for the reader. As a result, we have that $u_{k-1}^{\bar\ell+1}=u_{k-1}^{\bar\ell}$ and $x_k^{\bar\ell+1}=x_k^{\bar\ell}$, $k=1,\dots,N$. From the definition of the generator $\mathcal G$, cf. Eq.~\eqref{eq:generator_def}, we conclude that $\pi^{\bar\ell+1}=\pi^{\bar\ell}$. 
\end{proof}

\subsection{Simplified On-Line Policy Iteration}
One potential challenge in our algorithm is the minimization in Eqs.~\eqref{eq:pi_im_t1} and \eqref{eq:pi_im_tk}, which may not be easily done in real time. This typically occurs when the sets $U_k(x_k)$ are subsets of some Euclidean spaces, while the values $J_{k+1,\pi^{\ell}}(x_{k+1})$ can only be evaluated through simulation, one state at a time. In such cases, it is necessary to apply a simplified variant of the on-line PI algorithm, which we describe next. In this variant the constraint sets $U_k(x_k^{\ell+1})$ in the policy improvement operation \eqref{eq:pi_im_tk} are replaced by `simpler' and iteration-dependent sets $\overline U_k(x_k^{\ell+1},\mu_k^\ell)$, as we describe below. 

At the $\ell$th iteration, given a policy $\pi^\ell$, we obtain a new policy $\pi^{\ell+1}$ through the following steps:
\begin{itemize}
    \item[(1)] Starting with the fixed initial state, we construct a set $\overline U_0(x_0,\mu_0^\ell)$, which is a subset of $U_0(x_0)$. We then compute $u_0^{\ell+1}$ via the minimization
    \begin{equation}
        \label{eq:ssl_im_t1}
        u_0^{\ell+1}\in \arg\min_{u_0\in \overline U_0(x_0,\mu_0^\ell)}\Big[g_0(x_0,u_0)+J_{1,\pi^\ell}\big(f_0(x_0,u_0)\big)\Big],
    \end{equation}
    with preference given to $\mu_0^\ell(x_0)$ if it belongs to the set $\overline U_0(x_0,\mu_0^\ell)$ and attains the minimum above, and we set 
    \begin{equation}
        \label{eq:ssl_x1}
        x_1^{\ell+1}=f_0(x_0,u_0^{\ell+1}).
    \end{equation}
    \item[(2)] Going forward, for $k=1,\dots,N-1$, we construct  sets $\overline U_k(x_k^{\ell+1},\mu_k^\ell)$, which are subsets of $U_k(x_k^{\ell+1})$. We then compute $u_k^{\ell+1}$ via the minimization
    \begin{equation}
        \label{eq:ssl_im_tk}
        u_k^{\ell+1}\in \arg\min_{u_k\in \overline U_k(x_k^{\ell+1},\mu_k^\ell)}\Big[g_k(x_k^{\ell+1},u_k)+J_{k+1,\pi^{\ell}}\big(f_k(x_k^{\ell+1},u_k)\big)\Big],
    \end{equation}
    with preference given to $\mu_k^\ell(x_k^{\ell+1})$ if it belongs to the set $\overline U_k(x_k^{\ell+1},\mu_k^\ell)$ and attains the minimum above, and we set
    \begin{equation}
        \label{eq:ssl_xk}
        x_{k+1}^{\ell+1}=f_k(x_k^{\ell+1},u_k^{\ell+1}).
    \end{equation}
    \item[(3)] Finally, we obtain the new policy $\pi^{\ell+1}$ by using the generator $\mathcal G$, i.e., 
    \begin{equation}
        \label{eq:ssl_g}
        \pi^{\ell+1}=\mathcal{G}(x_0,u_0^{\ell+1},x_1^{\ell+1},\dots,u_{N-1}^{\ell+1},x_N^{\ell+1}).
    \end{equation}
\end{itemize} 

We call the preceding algorithm \emph{simplified on-line PI}, since the minimization uses the subsets $\overline U_0(x_0,\mu_0^\ell)$ and $\overline U_k(x_k^{\ell+1},\mu_k^\ell)$, of $U_0(x_0)$ and $U_k(x_k^{\ell+1})$, $k=1,\dots,N-1$, respectively. We can show that the desired cost improvement property is preserved, provided that these subsets satisfy the following condition:
\begin{equation}
    \label{eq:ssl_set}
    \mu_0^\ell(x_0)\in \overline U_0(x_0,\mu^\ell_0),\quad \mu_k^\ell(x_k^{\ell+1})\in \overline U_k(x_k^{\ell+1},\mu^\ell_k),\qquad k=1,\dots,N-1.
\end{equation}
Intuitively, this condition ensures that when the controls $u_k^{\ell+1}$, $k=0,\dots,N-1$ are selected via the minimizations \eqref{eq:ssl_im_t1} and \eqref{eq:ssl_im_tk}, they are compared against the controls selected by $\mu_k^\ell$. The result is provided in the following proposition. Its proof is nearly identical to that of Prop.~\ref{prop:sl_property} and is left for the reader.

\begin{proposition}[Properties of Simplified On-Line PI]\label{prop:ssl_property}
    Consider the sequence of policies $\pi^0,\pi^1,\dots$, computed via the simplified on-line PI algorithm \eqref{eq:ssl_im_t1}-\eqref{eq:ssl_g}. Suppose that the generator is consistent and that the subsets $\overline U_0(x_0,\mu_0^\ell)$, $\overline U_k(x_k^{\ell+1},\mu_k^\ell)$, $k=1,\dots,N-1$, satisfy condition \eqref{eq:ssl_set}.
    \begin{itemize}
        \item[(a)] The policies are improving at $x_0$ in the sense that 
    \begin{equation}
        \label{eq:ssl_imp}
        J_{\pi^{\ell+1}}(x_0)\leq J_{\pi^{\ell}}(x_0), \qquad \hbox{for all $\ell$}, 
    \end{equation}
    and the sequence $\{J_{\pi^\ell}(x_0)\}$ is convergent.
    \item[(b)] If equality in \eqref{eq:ssl_imp} holds for some $\bar \ell$, we have $\pi^{\ell+1}=\pi^\ell$ for all $\ell\geq \bar \ell\geq 1$.
    \end{itemize}
\end{proposition}

Compared with the standard on-line PI algorithm, the simplified version requires the computation of the sets $\overline{U}_k(x_k,\mu_k)$, $k=0,\dots,N-1$, which satisfy the condition \eqref{eq:ssl_set}. While it may not be clear how such sets can be constructed in general, there is a principled and very effective method to do so in the case where the control $u_k$ is composed of multiple components (the multiagent case). We discuss this case next, and we show how to compute the sets $\overline{U}_k(x_k,\mu_k)$ to simplify greatly the minimizations in Eqs.~\eqref{eq:ssl_im_t1} and \eqref{eq:ssl_im_tk}.

\subsection{Multiagent On-Line Policy Iteration}
Suppose that the control $u_k$ is composed of multiple components. Our on-line PI algorithm for this case treats each component of $u_k$ as an `agent,' and performs minimization `one agent at a time,' either approximately or exactly. This minimization approach was proposed and justified mathematically by the last author in the papers \cite{bertsekas2020multiagent,bertsekas2021multiagent}.\footnote{Note that the term `agent' in our terminology refers to a coordinate (component) of the control vector used for decomposition of the optimization problem. This differs from some control and multiagent systems literature (e.g., \cite{olfati2007consensus,dimarogonas2011distributed,chen2024cooperative}), where an `agent' denotes an independent decision-making entity, often corresponding to a physical subsystem.} We first consider the case where the control components are continuous, so the minimization for each `agent' can only be done approximately. Then we comment on the case where the control components are discrete.   

To this end, let us assume that for all $x_k$ and $k$, \emph{the control constraint set $U_k(x_k)$ is a subset of an $n$-dimensional Euclidean space}. We consequently write the control $u_k$ as $(u_{k,1},\dots,u_{k,n})$, where $u_{k,i}$ is the $i$th component of $u_k$ (which is viewed as an agent, as noted earlier). Similarly, given a policy $\pi=\{\mu_0,\dots,\mu_{N-1}\}$, we denote by $\mu_{k,i}(x_k)$ the $i$th component of $\mu_k(x_k)$. 

We illustrate the key ideas of the construction for the case of two agents ($n=2$) in Fig.~\ref{fig:simple_pi}. At state $x_k^{\ell+1}$ during iteration $(\ell+1)$, the  constraint set $U_k(x_k^{\ell+1})$ is shown as the green ellipsoid. Starting from the control $\mu_k^{\ell}(x_k^{\ell+1})$ (shown as a black square), we first construct a set $\hat U_{k,1}(x_k^{\ell+1},\mu_k^\ell)$, which consists of all the feasible controls whose second component is held fixed at $\mu_{k,2}^\ell(x_k^{\ell+1})$, while the first component takes some discrete values (red dots along the horizontal line), including $\mu_{k,1}^\ell(x_k^{\ell+1})$. We minimize some function (to be specified later) over this set to obtain $u_{k,1}^{\ell+1}$, yielding the intermediate point $\big(u_{k,1}^{\ell+1},\mu_{k,2}^\ell(x_k^{\ell+1})\big)$ (black dot). Next, we construct a set $\hat U_{k,2}(x_k^{\ell+1},\mu_k^\ell)$ by fixing the first component at the newly optimized $u_{k,1}^{\ell+1}$ and varying the second component over some discrete values (blue dots along the vertical line), including $\mu_{k,2}^\ell(x_k^{\ell+1})$. Minimizing over this set determines $u_{k,2}^{\ell+1}$, giving the final control $u_k^{\ell+1}$ (black star). The procedure generalizes to $n$ components in an obvious way.

\begin{figure}[t]
    \centering
    \includegraphics[width=\linewidth]{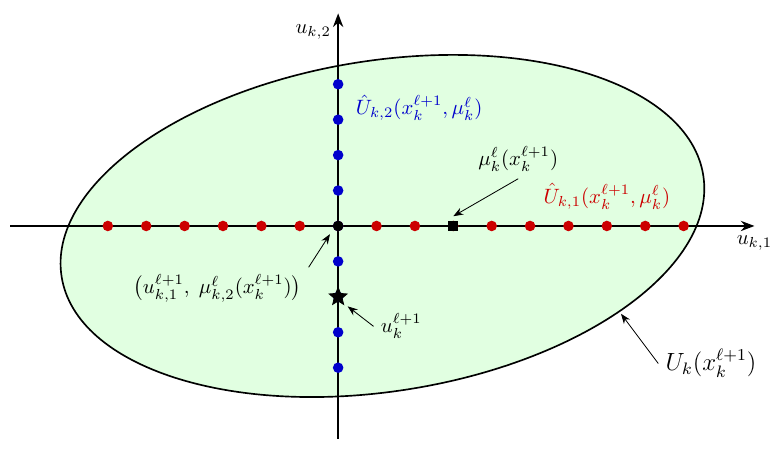}
    \caption{Illustration of the coordinate-wise discretization construction in simplified on-line PI for $n=2$. The green region is the constraint set $U_k(x_k^{\ell+1})$. Starting from the control $\mu_k^\ell(x_k^{\ell+1})$ (black square), we first minimize over the discretized grid along the first coordinate [red dots, forming the set $\hat U_{k,1}(x_k^{\ell+1},\mu_k^\ell)$] to obtain the intermediate point $(u_{k,1}^{\ell+1},\, \mu_{k,2}^\ell(x_k^{\ell+1}))$ (black circle). We then minimize over the grid along the second coordinate [blue dots, forming the set $\hat U_{k,2}(x_k^{\ell+1},\mu_k^\ell)$], yielding the final control $u_k^{\ell+1}$ (black star).}
    \label{fig:simple_pi}
\end{figure}

In particular, let $\rho_i$, $i=1,\dots,n$, be some selectable constants, which define discretization resolutions, and $Z$ be a bounded subset of integers that includes $0$, which is used to define the discretization points.\footnote{Our multiagent PI can be extended to the case where the constants $\rho_i$, $i=1,\dots,n$ and/or the set $Z$ vary with the stage $k$ and/or the iteration number $\ell$. We keep them constant for simplicity.} Multiagent on-line PI proceeds as simplified on-line PI, except that the single minimization in Eqs.~\eqref{eq:ssl_im_t1} and \eqref{eq:ssl_im_tk} is replaced by a series of minimizations. Specifically, during iteration $(\ell+1)$, at state $x_k^{\ell+1}$, we first construct a set $\hat U_{k,1}(x_k^{\ell+1},\mu_k^\ell)$ of the form
\begin{equation}
    \label{eq:hat_u_1}
    \begin{aligned}
        \hat U_{k,1}(x_k^{\ell+1},\mu_k^\ell)=&\Big\{u_k\,\big|\,u_{k,1}=\mu_{k,1}^\ell(x_k^{\ell+1})+m\rho_1, m\in Z,\\
    & \qquad \;\,u_{k,i}=\mu_{k,i}^\ell(x_k^{\ell+1}),\,i=2,\dots,n,\,u_k\in  U_{k}(x_k^{\ell+1})\Big\}.
    \end{aligned}
\end{equation}
We compute a control $\hat u_{k}$ via the minimization
\begin{equation}
    \label{eq:hat_u_min_1}
    \hat u_k\in \arg\min_{u_k\in \hat U_{k,1}(x_k^{\ell+1},\mu_k^\ell)}\Big[g_k(x_k^{\ell+1},u_k)+J_{k+1,\pi^{\ell}}\big(f_k(x_k^{\ell+1},u_k)\big)\Big],
\end{equation}
with preference given to $\mu_k^\ell(x_k^{\ell+1})$ in case of a tie, and set $u_{k,1}^{\ell+1}=\hat u_{k,1}$. Having computed $u_{k,1}^{\ell+1},\dots,u_{k,i-1}^{\ell+1}$, $i=2,\dots,n-1$, we construct the set $\hat U_{k,i}(x_k^{\ell+1},\mu_k^\ell)$ as
\begin{equation}
    \label{eq:hat_u_i}
    \begin{aligned}
        \hat U_{k,i}(x_k^{\ell+1},\mu_k^\ell)=&\Big\{u_k\,\big|\,u_{k,j}=u_{k,j}^{\ell+1},\,j=1,\dots,i-1,\\
        &\qquad \;\,u_{k,i}=\mu_{k,i}^\ell(x_k^{\ell+1})+m\rho_i, m\in Z,\\
    & \qquad \;\,u_{k,j}=\mu_{k,j}^\ell(x_k^{\ell+1}),\,j=i+1,\dots,n,\,u_k\in  U_{k}(x_k^{\ell+1})\Big\},
    \end{aligned}
\end{equation}
we compute $\hat u_{k}$ via the minimization
\begin{equation}
    \label{eq:hat_u_min_k}
    \hat u_k\in \arg\min_{u_k\in \hat U_{k,i}(x_k^{\ell+1},\mu_k^\ell)}\Big[g_k(x_k^{\ell+1},u_k)+J_{k+1,\pi^{\ell}}\big(f_k(x_k^{\ell+1},u_k)\big)\Big],
\end{equation}
with preference given to $\mu_k^\ell(x_k^{\ell+1})$ in case of a tie, and we set $u_{k,i}^{\ell+1}=\hat u_{k,i}$. Finally, after obtaining $u_{k,i}^{\ell+1}$, $i=1,\dots,n-1$, we construct the set $\hat U_{k,n}(x_k^{\ell+1},\mu_k^\ell)$ as
\begin{equation}
    \label{eq:hat_u_n}
    \begin{aligned}
        \hat U_{k,n}(x_k^{\ell+1},\mu_k^\ell)=&\Big\{u_k\,\big|\,u_{k,j}=u_{k,j}^{\ell+1},\,j=1,\dots,n-1,\\
        &\qquad \;\,u_{k,n}=\mu_{k,n}^\ell(x_k^{\ell+1})+m\rho_n, m\in Z,\,u_k\in  U_{k}(x_k^{\ell+1})\Big\},
    \end{aligned}
\end{equation}
compute $\hat u_{k}$ via the minimization
\begin{equation}
    \label{eq:hat_u_min_n}
    \hat u_k\in \arg\min_{u_k\in \hat U_{k,n}(x_k^{\ell+1},\mu_k^\ell)}\Big[g_k(x_k^{\ell+1},u_k)+J_{k+1,\pi^{\ell}}\big(f_k(x_k^{\ell+1},u_k)\big)\Big],
\end{equation}
with preference given to $\mu_k^\ell(x_k^{\ell+1})$ in case of a tie, and set $u_{k,n}^{\ell+1}=\hat u_{k,n}$. 

It can be seen that the sets $\hat{U}_{k,i}(x_k^{\ell+1},\mu_k^\ell)$, $i=1,\dots,n$, are finite (since $Z$ is bounded). As a result, the minimizations in Eqs.~\eqref{eq:hat_u_min_1}, \eqref{eq:hat_u_min_k}, and \eqref{eq:hat_u_min_n} involve computation and comparison of finitely many values, which are well-suited for parallel computation. 

As mentioned earlier, multiagent on-line PI can be viewed as a special case of the simplified on-line PI discussed earlier. To see this, we set
$$\overline U_k(x_k^{\ell+1},\mu_k^\ell)=\cup_{i=1}^n \hat U_{k,i}(x_k^{\ell+1},\mu_k^\ell).$$
We have that the control $u_k^{\ell+1}$ computed via Eqs.~\eqref{eq:hat_u_1}-\eqref{eq:hat_u_min_n} attains the minimum in
$$\min_{u_k\in \overline U_k(x_k^{\ell+1},\mu_k^\ell)}\Big[g_k(x_k^{\ell+1},u_k)+J_{k+1,\pi^{\ell}}\big(f_k(x_k^{\ell+1},u_k)\big)\Big],$$
cf. Eq.~\eqref{eq:ssl_im_tk}, and equals $\mu_k^\ell(x_k^{\ell+1})$ in case of a tie. Moreover, we have $\mu_k^\ell(x_k^{\ell+1})\in \hat U_{k,1}(x_k^{\ell+1},\mu_k^\ell)\subset \overline U_k(x_k^{\ell+1},\mu_k^\ell)$. As a result, Prop.~\ref{prop:ssl_property} applies to the multiagent on-line PI algorithm.

On the other hand, multiagent on-line PI results in tremendous computational savings, especially when discretization is involved. Suppose that the discretization resolutions $\rho_i$, $i=1,\dots,n$, are selected so that each set $\hat U_{k,i}(x_k^{\ell+1},\mu_k^\ell)$ consists of $m$ components. Then the minimizations in Eqs.~\eqref{eq:hat_u_min_1}, \eqref{eq:hat_u_min_k}, and \eqref{eq:hat_u_min_n} compute and compare a total of $m\cdot n$ numbers. By contrast, if we construct the set $\overline U_k(x_k^{\ell+1},\mu_k^\ell)$ by using the same resolutions $\rho_i$, $i=1,\dots,n$ along different coordinates and discretizing the set $U_k(x_k^{\ell+1})$ uniformly, the resulting set $\overline U_k(x_k^{\ell+1},\mu_k^\ell)$ contains $m^n$ elements, which scales poorly with the dimension of the control $u_k$.

The preceding construction of the sets $\overline U_k(x_k^{\ell+1},\mu_k^\ell)$ applies to the case where the components of $u_k$ are continuous, and for this reason we have used discretization. A similar procedure applies in the case where the sets $U_k(x_k)$ are finite, or involve both continuous and discrete components. Then multiagent on-line PI is still a special case of the simplified on-line PI algorithm, and Prop.~\ref{prop:ssl_property} holds in this case as well.

We note that coupling constraints among control components may restrict the possibility of improving the controls through coordinate-wise minimization, even though feasibility is always preserved. One possible remedy is to optimize groups of coupled control components simultaneously, as discussed in \cite[Section~3.2.2]{bertsekas2020rollout}. Such an extension preserves the cost-improvement property, provided that the restricted control sets satisfy condition \eqref{eq:ssl_set}.

\subsection{Stochastic Variants of On-Line Policy Iteration}
Up to this point, we have described our on-line PI algorithm and its simplified variant as deterministic algorithms. In particular, the generator $\mathcal G$ is defined as a function that maps a complete feasible trajectory to a policy. Similarly, in simplified on-line PI, for $k=0,\dots,N-1$, the set-valued function $\overline U_k(\cdot,\cdot)$ maps $(x_k,\mu_k)$ to a small subset of $U_k(x_k)$. 

In practice, assuming that either the generator $\mathcal{G}$ or the sets $\overline U_k(x_k,\mu_k)$ are deterministically chosen may be restrictive. For example, given a complete feasible trajectory, one may wish to use random sampling and stochastic training algorithms to obtain a policy $\pi$, in which case the generator $\mathcal{G}$ is defined implicitly through these procedures. Similarly, one may construct the set $\overline U_k(x_k,\mu_k)$ through stochastic sampling schemes. In this section we show how our on-line PI algorithm and its simplified variant extend to the case where $\mathcal{G}$ and/or $\overline U_k(x_k,\mu_k)$ are chosen stochastically.\footnote{Here stochasticity refers to the random selection of the policy generator and/or the restricted control sets, while the dynamic system and the stage costs remain deterministic. Extensions to stochastic problems require additional analysis, since consistency along a single realized trajectory does not in general guarantee improvement of expected policy costs.}

In particular, given a feasible trajectory $\{x_0,u_0,x_1,\dots,u_{N-1},x_N\}$, a stochastic generator $\tilde{\mathcal{G}}$ defines a probability distribution over the set of all policies $\Pi$.\footnote{To bypass measure-theoretic issues, we assume that this probability distribution is discrete.} Then a policy $\pi$ can be sampled from this distribution, which we write as
\begin{equation}
    \label{eq:generator_def_st}
    \pi\sim \tilde{\mathcal{G}}(x_0,u_0,x_1,\dots,u_{N-1},x_N).
\end{equation}
We say that the stochastic generator $\tilde{\mathcal{G}}$ is \emph{consistent} if $\pi$ satisfies
$$\mu_k(x_k)=u_k,\qquad k=0,\dots,N-1;$$
cf. Eq.~\eqref{eq:consistency}. We have the following policy improvement property, which is just Prop.~\ref{prop:sl_property}(a) adapted to stochastic on-line PI. Its proof can be obtained by using the arguments of the proof of Prop.~\ref{prop:sl_property}(a).

\begin{proposition}[Convergence of Stochastic On-Line PI]\label{prop:sl_st_converg}
    Consider the sequence of policies $\pi^0,\pi^1,\dots$, computed via the on-line PI algorithm \eqref{eq:pi_im_t1}-\eqref{eq:sl_g} with Eq.~\eqref{eq:sl_g} replaced by
    $$\pi^{\ell+1}\sim \tilde{\mathcal{G}}(x_0,u_0^{\ell+1},x_1^{\ell+1},\dots,u_{N-1}^{\ell+1},x_N^{\ell+1}).$$
    Suppose that the generator $\tilde{\mathcal{G}}$ is consistent. Then the generated policies are improving at $x_0$ in the sense that we have
    $$J_{\pi^{\ell+1}}(x_0)\leq J_{\pi^{\ell}}(x_0), \qquad \hbox{for all $\ell$}.$$   
\end{proposition}

We note, however, that Prop.~\ref{prop:sl_property}(b) does not extend to the stochastic on-line PI. To see this, suppose that for some $\bar \ell$, we have $J_{\pi^{\bar\ell+1}}(x_0)= J_{\pi^{\bar\ell}}(x_0)$. Applying the arguments in the proof of Prop.~\ref{prop:sl_property}(b) yields that 
$$u_{k-1}^{\bar\ell+1}= u_{k-1}^{\bar\ell},\quad x_k^{\bar\ell+1}= x_k^{\bar\ell},\qquad k=1,\dots,N.$$
However, since a stochastic generator $\tilde{\mathcal{G}}$ is applied, we do not have $\pi^{\bar \ell+1}=\pi^{\bar \ell}$ in general. Still, this can potentially be beneficial as it helps the algorithm step away from a policy $\pi^{\bar \ell}$ that may  be poor.

Similarly, in the simplified on-line PI algorithm \eqref{eq:ssl_im_t1}-\eqref{eq:ssl_g}, the sets $\overline{U}_0(x_0,\mu_0^{\ell})$, $\overline{U}_k(x_k^{\ell+1},\mu_k^{\ell})$, $k=1,\dots,N-1$, can be replaced by sets $\tilde{U}_0(x_0,\mu_0^{\ell})$, $\tilde{U}_k(x_k^{\ell+1},\mu_k^{\ell})$, $k=1,\dots,N-1$, which are sampled from some discrete distributions that depend on $(x_0,\mu_0^{\ell})$ and $(x_k^{\ell+1},\mu_k^{\ell})$, respectively. As in the simplified PI algorithm, we require that these sets satisfy the condition
\begin{equation}
    \label{eq:ssl_st_set}
    \mu_0^\ell(x_0)\in \tilde U_0(x_0,\mu_0^\ell),\quad \mu_k^\ell(x_k^{\ell+1})\in \tilde U_k(x_k^{\ell+1},\mu_k^\ell),\qquad k=1,\dots,N-1,
\end{equation}
cf. Eq.~\eqref{eq:ssl_set}. Then the following result is obtained.
\begin{proposition}[Convergence of Stochastic Simplified On-Line PI]\label{prop:ssl_st_converg}
Consider the sequence of policies $\pi^0,\pi^1,\dots$, computed via the simplified on-line PI algorithm \eqref{eq:ssl_im_t1}-\eqref{eq:ssl_g} with the sets $\overline U_0(x_0,\mu_0^\ell)$, $\overline U_k(x_k^{\ell+1},\mu_k^\ell)$, $k=1,\dots,N-1$, replaced by the randomly generated sets $\tilde{U}_0(x_0,\mu_0^{\ell})$, $\tilde{U}_k(x_k^{\ell+1},\mu_k^{\ell})$, $k=1,\dots,N-1$. Assume that these sets satisfy condition \eqref{eq:ssl_st_set}. Suppose that either
    \begin{itemize}
        \item[(a)] The generator function $\mathcal G$ is consistent, or
        \item[(b)] Eq.~\eqref{eq:ssl_g} is replaced by
    $$\pi^{\ell+1}\sim \tilde{\mathcal{G}}(x_0,u_0^{\ell+1},x_1^{\ell+1},\dots,u_{N-1}^{\ell+1},x_N^{\ell+1}),$$
    and the generator $\tilde{\mathcal{G}}$ is consistent.
    \end{itemize}
    Then the generated policies are improving at $x_0$ in the sense that we have
    $$J_{\pi^{\ell+1}}(x_0)\leq J_{\pi^{\ell}}(x_0), \qquad \hbox{for all $\ell$}.$$    
\end{proposition}

The proof can be obtained by adapting the arguments in the proof of Prop.~\ref{prop:sl_property}(a). Regarding Prop.~\ref{prop:sl_property}(b), it does not extend to the stochastic version of simplified on-line PI, even if the generator function $\mathcal G$ is used. To see this, note that if $J_{\pi^{\bar\ell+1}}(x_0)= J_{\pi^{\bar\ell}}(x_0)$ for some $\bar \ell$, we can show that $\pi^{\bar\ell+1}=\pi^{\bar\ell}$ and the trajectory computed at the $(\bar\ell+1)$th iteration is the same as that at the $\bar\ell$th iteration. However, the set $\tilde{U}_0(x_0,\mu_0^{\bar \ell+1})$ computed at the first step in the $(\bar \ell+2)$th iteration may be different from $\tilde{U}_0(x_0,\mu_0^{\bar \ell})$, despite the fact that $\mu_0^{\bar \ell+1}=\mu_0^{\bar \ell}$.

For an example of a stochastic mechanism for generating a set $\tilde{U}_k(x_k^{\ell+1},\mu_k^{\ell})$, consider the case where $U_k(x_k)$ is a subset of an $n$-dimensional Euclidean space. Instead of computing sets $\hat{U}_{k,i}(x_k^{\ell+1},\mu_k^\ell)$, $i=1,\dots,n$, sequentially according to a fixed order via Eqs.~\eqref{eq:hat_u_1}-\eqref{eq:hat_u_min_n}, and defining the set $\overline{U}_k(x_k^{\ell+1},\mu_k^{\ell})$ as $\cup_{i=1}^n\hat{U}_{k,i}(x_k^{\ell+1},\mu_k^\ell)$, we may first generate a random permutation function $\sigma:\{1,\dots,n\}\mapsto\{1,\dots,n\}$, compute the sets $\hat{U}_{k,i}(x_k^{\ell+1},\mu_k^\ell)$, $i=1,\dots,n$, according to the order $\sigma(i)$, and define the set $\tilde{U}_k(x_k^{\ell+1},\mu_k^{\ell})$ as the union of these sets. Similar to the use of a stochastic generator $\tilde{\mathcal{G}}$, this stochastic mechanism may be helpful in some contexts. In fact, a computational study given in \cite{emanuelsson2023multiagent} has shown the effectiveness of this mechanism when applied in rollout, i.e., one iteration of PI.

\section{Computational Studies}\label{sec:compute}

Two principal classes of problems that are well suited for our algorithm are  discrete optimization (including integer programming) and path planning involving tasks that are solved repetitively, starting from the same initial state. In this section, we apply our on-line PI algorithm and its stochastic multiagent variant to representative problems from these two categories. The first is a multidimensional assignment problem, the second is path planning for robotic tasks involving obstacles, the third is path planning and coordination for robotic tasks with multiple drones, and the fourth is path planning for a robot arm. We show that within these contexts, the generator $\mathcal G$ or $\tilde{\mathcal{G}}$ can be properly defined and our on-line PI algorithm is very effective. Finally, we also compare our method with the scheme introduced in \cite{rosolia2017learning} at the end of this section. The code can be found at \url{https://github.com/yuchaotaigu/On-Line-Policy-Iteration}.

\subsection{Multidimensional Assignment}

Our first computational experiments involve small integer programming problems, and are provided here to illustrate our methodology within a simple setting where the optimal solution can be calculated exactly. We address multidimensional assignment, an important class of discrete optimization problems that arise in a variety of application contexts; see, e.g., \cite{kammerdiner2022multidimensional,musunuru2024approximate}, and the references cited there. 

Mathematically, the problem involves a graph consisting of $(N+1)$ subsets of nodes $\mathcal N_0,\dots,\mathcal N_{N}$, which are referred to as \emph{layers}. Each layer consists of $m$ nodes; see Fig.~\ref{fig:mda_eg}. The arcs of the graph are of the form $(i,j)$, where $i$ is a node in a layer $\mathcal N_k$, $k=0,1,\dots,N-1$, and $j$ is a node in the corresponding next layer $\mathcal{N}_{k+1}$. We consider subsets of $(N+1)$ nodes, referred to as \emph{groupings}, which consist of a single node from every layer. For each grouping, there is an associated cost. A partition of the set of nodes into $m$ disjoint groupings is called an $(N+1)$-dimensional assignment, and its cost is the sum of the costs of its $m$ groupings. The problem is to find an $(N+1)$-dimensional assignment of minimum cost. We assume that the graph is dense in the sense that for every node $i$ in layer $\mathcal N_k$ and every node $j$ in layer $\mathcal N_{k+1}$, $(i,j)$ is an arc.\footnote{This is without loss of generality because if some arcs are absent in a given graph, we can add these arcs to the graph and assign high costs to the groupings that involve these arcs.} The difficulty here is that the cost of a grouping does not decompose into a sum of its $N$ arc costs, so the problem cannot be solved by solving $N$ decoupled $2$-dimensional assignment problems.

\begin{figure}[t]
    \centering
    \includegraphics[width=0.7\linewidth]{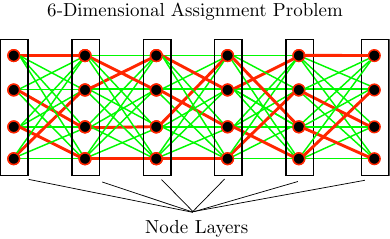}
    \caption{Illustration of the graph of an $(N+1)$-dimensional assignment problem (here $N=5$). There are $N+1$ node layers each consisting of $m$ nodes (here $m=4$). Each grouping consists of $(N+1)$ nodes, one from each layer, and $N$ corresponding arcs. An $(N+1)$-dimensional assignment consists of $m$ node-disjoint groupings, where each node belongs to one and only one grouping (illustrated in the figure with thick red lines). For each grouping, there is an associated cost that depends on the $N$-tuple of arcs comprising the grouping. The cost of an $(N+1)$-dimensional assignment is the sum of the costs of its $m$ groupings.}
    \label{fig:mda_eg}
\end{figure}

The problem can be formulated as an $N$-stage optimal control problem where at the $k$th stage, we select the arcs of a two-dimensional assignment
$$u_k=\{(i_n,j_n)\,|\,n=1,\dots,m\},\quad k=0,\dots,N-1,$$
which connect the nodes $i_n$ of layer $\mathcal{N}_k$ and the nodes $j_n$ of layer $\mathcal{N}_{k+1}$ on a one-to-one basis. The initial state $x_0$ is an artificial state, the next state $x_1$ after applying $u_0$ is $x_1=u_0$, and the state equation is 
$$x_{k+1}=(x_k,u_k),\qquad k=1,\dots,N-1.$$
Thus a state $x_k$ consists of the first $k$ arcs of $m$ groupings. The control constraint set $U_k(x_k)$ is independent of $x_k$ and is the set of the two-dimensional arc assignments between layers $\mathcal N_k$ and $\mathcal N_{k+1}$. There is a terminal cost $g_N(x_N)$ (the sum of the costs of the $m$ groupings in the terminal state $x_N$), while all the other stage costs are $0$. Further discussion can be found in \cite[Example~2.2.1]{bertsekas2025course}, which addresses the solution of multiassignment problems with alternative RL methods.

For a given trajectory $x_0,u_0^\ell,x_1^\ell,\dots,u_{N-1}^\ell, x_N^\ell$, a generator defines a policy $\pi^\ell=\{\mu_0^\ell,\dots,\mu_{N-1}^\ell\}$ such that $\mu_k^\ell(x_k)=u_k^\ell$, i.e., the policy adds the arcs of the two-dimensional assignment $u_k^\ell$ to the partial grouping/trajectory $x_k$. It can be seen that the generator is consistent. For  this generator, and given the current trajectory, on-line PI constructs a new trajectory by solving sequentially $N$ two-dimensional assignment problems: first a problem involving nodes of layers $\mathcal{N}_0$ and $\mathcal{N}_1$, then a problem involving nodes of layers $\mathcal{N}_1$ and $\mathcal{N}_2$, etc. The cost of an arc $(i,j)$ in each two-dimensional assignment problem is provided by the cost of the grouping formed by $(i,j)$, the arcs of the new trajectory leading to $i$, and the arcs of the current trajectory starting from $j$; see \cite[Example~2.2.1]{bertsekas2025course} for further discussion.

Figure~\ref{fig:mda_gap} demonstrates the computational results of on-line PI applied to some small-scale problems, where the exact solution can be computed via brute force. On-line PI takes about $2$ milliseconds, while exact optimization by brute-force calculation takes up to $3$ seconds. In all these examples, the initial trajectory $x_0,u_0^0,x_1^0,\dots,u_{N-1}^0,x_N^0$ was  generated randomly, and  we have defined the initial policy $\pi^0=\{\mu_0^0,\dots,\mu_{N-1}^0\}$ via $\mu_k^0(x_k)\equiv u_k^0$ for all $k$. Generally, in integer programming problems, one can often use various heuristics to generate an initial trajectory. One may  then try to improve this trajectory by using rollout techniques, similar to the methodology of this paper; see the book \cite{bertsekas2025course} for further discussion. We have not assumed any special structure in our experiments.

Figure~\ref{fig:mda_imp} shows the results of on-line PI applied to larger problems with $N$ ranging from $8$ to $12$ and $m$ ranging from $3$ to $6$, where exact optimization by brute-force calculation is not feasible.\footnote{To provide a sense of problem size, consider $N=12$ and $m=6$. Suppose the cost of each grouping is stored using a single $8$-bit byte. Then we need over $100$ GB to store the cost values of all groupings in the problem.} On-line PI generates cost improving solutions in all these cases, with computation time less than $70$ milliseconds in the worst case.

\begin{figure}[t]
    \centering
    \includegraphics[width=\linewidth]{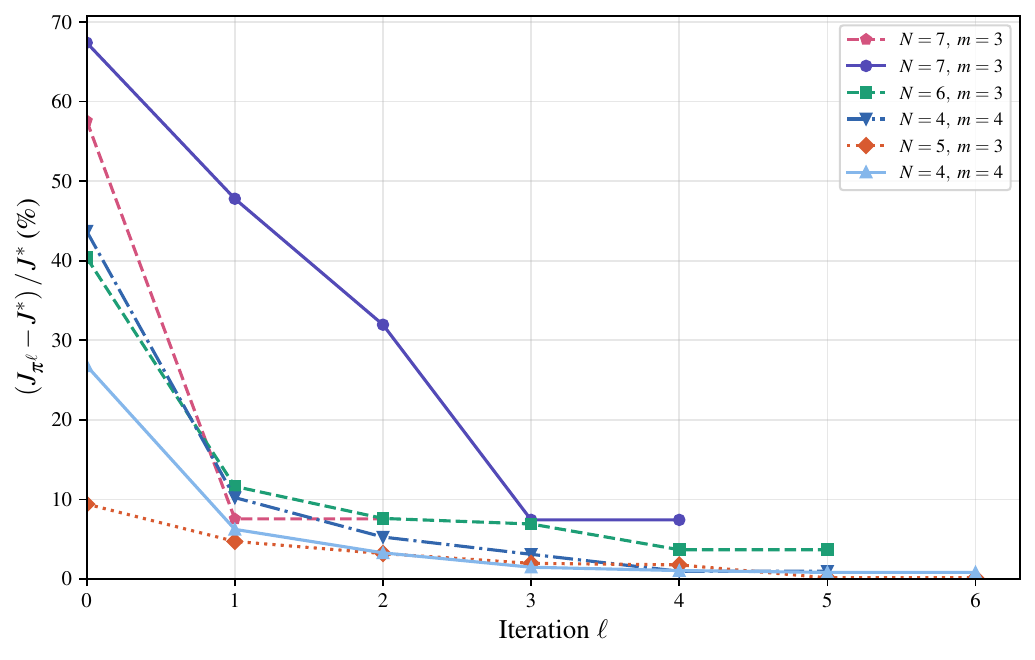}
    \caption{Optimality gap $(J_{\pi^\ell}-J^*)/J^*$ $(\%)$ for the on-line PI algorithm applied to some multidimensional assignment problems with randomly generated grouping costs. It can be seen that on-line PI converges monotonically from $2$ to $6$ iterations and significantly reduces the optimality gap in all of these examples.}
    \label{fig:mda_gap}
\end{figure}

\begin{figure}[ht]
    \centering
    \includegraphics[width=\linewidth]{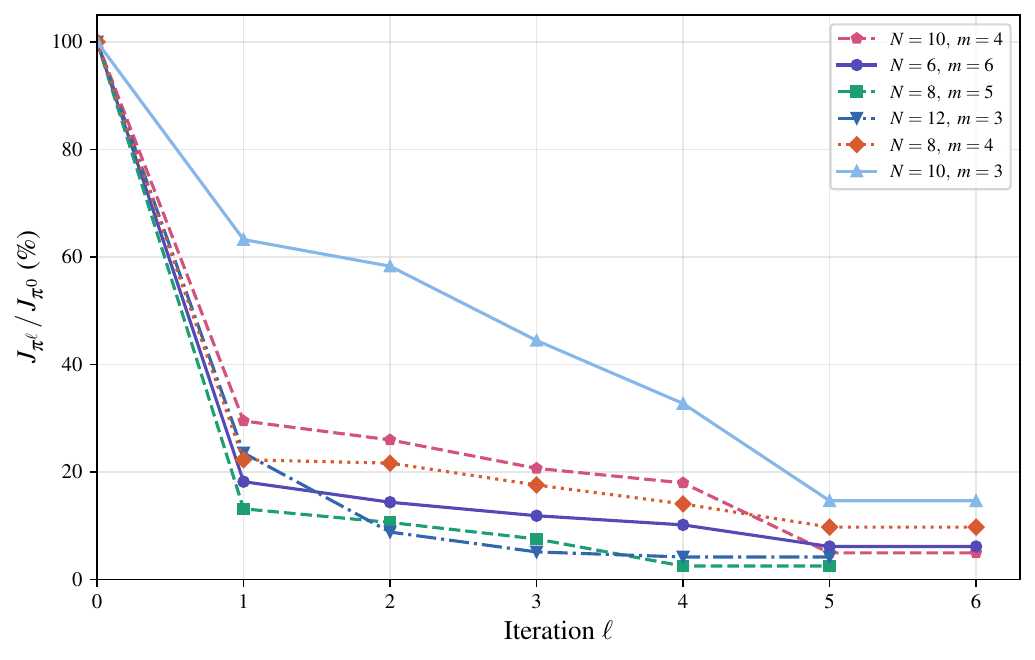}
    \caption{Cost improvement $J_{\pi^\ell}/J_{\pi^0}$ $(\%)$ for the on-line PI algorithm applied to some multidimensional assignment problems with randomly generated grouping costs, where the exact solution is intractable via brute force. The on-line PI algorithm reduces the costs in all tested cases.}
    \label{fig:mda_imp}
\end{figure}

\subsection{Planning the Path of a Drone}\label{sec:4_2}
We consider a path planning problem for a drone in a $3$D environment. The states $x_k$ are $6$D vectors, consisting of positions and velocities in $3$ dimensions. The controls $u_k$ are $3$D vectors, representing acceleration in $3$ dimensions. The control constraint sets are state-independent, i.e., $U_k(x_k)\equiv U_k$ for all $x_k$ and $k$, and are characterized by upper and lower bounds on the accelerations. The state equation $f_k$ is a double integrator, and the stage cost $g_k$ penalizes control effort and introduces obstacle repulsion, as used in~\cite{khatib1986real}. The terminal cost $g_N$ is defined to penalize deviations from the goal state. The number of stages $N$ is $40$, and the time between two stages $k$ and $k+1$ is $0.25$ seconds, which is a standard time limit for computation of control in real-time operation. The overall objective is to safely reach the goal region within $N$ stages while avoiding obstacles. Further details on problem data can be found in the provided code.

We have applied our multiagent on-line PI algorithm to this problem, computing acceleration in one direction at a time (so in our terminology there are $3$ `agents'). In particular, our initial policy $\pi^0$ is obtained by using the proximal policy optimization algorithm (PPO) \cite{schulman2017proximal}, with the rewards defined as the negatives of the corresponding stage and terminal costs. The critical challenge is to construct a generator function $\mathcal G$ or a stochastic generator $\tilde{\mathcal{G}}$ while ensuring its consistency. In this example, we construct a stochastic generator $\tilde{\mathcal{G}}$ through the following three steps. At iteration $(\ell+1)$, given the feasible trajectory $x_0,u_0^{\ell+1},x_1^{\ell+1},\dots,u_{N-1}^{\ell+1},x_N^{\ell+1}$:
\begin{itemize}
    \item[1)] For every $x_k^{\ell+1}$, $k=1,\dots,N-1$, we sample $q_k$ states $\bar x_k^{j}$, $j=1,\dots,q_k$, according to a Gaussian distribution with mean equal to $x_k^{\ell+1}$ ;
    \item[2)] For every sampled state $\bar x_k^{j}$, $j=1,\dots,q_k$, and $k=1,\dots,N-1$, we compute a `good control' $\bar u_k^{j}$, $j=1,\dots,q_k$, and $k=1,\dots,N-1$. We thus obtain a data set 
    \begin{equation}
    \label{eq:train_set}
        \big\{(\bar x_k^{j},\bar u_k^{j})\,\big|\,k=1,\dots,N-1,\,j=1,\dots,q_k\big\};
    \end{equation}
    \item[3)] Finally, we use the  data set \eqref{eq:train_set} and  supervised learning to train a neural network $\mathcal F_\theta$ parametrized by $\theta$.  
\end{itemize}

These three steps collectively define the generator $\tilde{\mathcal{G}}$, whereby using the trajectory $x_0,u_0^{\ell+1},x_1^{\ell+1},\dots,u_{N-1}^{\ell+1},x_N^{\ell+1}$ as a starting point, we obtain the trained neural network $\mathcal F_\theta$ to define the policy $\pi^{\ell+1}$. In what follows, we elaborate on the computation of the `good control' $\bar u_k^{j}$ in step 2), and we show the consistency property of the stochastic generator $\tilde{\mathcal{G}}$.

For the computation of $\bar u_k^{j}$ in step 2), we first note that because the control constraint set  is state-independent, we can `replay' the control of the `current' trajectory. In particular, given $x_0,u_0^{\ell+1},x_1^{\ell+1},\dots,u_{N-1}^{\ell+1},x_N^{\ell+1}$, we define an intermediate policy $\bar \pi^{\ell+1}=\{\bar\mu_0^{\ell+1},\dots,\bar\mu_{N-1}^{\ell+1}\}$ such that $\bar \mu_k^{\ell+1}(x_k)\equiv u_k^{\ell+1}$. For every sampled state $\bar x_k^j$, as per step 1) above, we sample $n$ controls $\bar u_k^{j,s}$, $s=1,\dots,n$, according to a Gaussian distribution with mean equal to $u_k^{\ell+1}$ while ensuring that the control constraints are satisfied. We use the intermediate policy $\bar \pi^{\ell+1}$ to evaluate these randomly selected controls $\bar u_k^{j,s}$, and obtain the `good control' $\bar u_k^j$ via
$$\bar u_k^j\in \arg\min_{s=1,\dots,n}\Big[g_k(\bar x_k^{j},\bar u_k^{j,s})+J_{k+1,\bar \pi^{\ell+1}}\big(f_k(\bar x_k^{j},\bar u_k^{j,s})\big)\Big].$$
This is a straightforward calculation that can be carried out easily.

To show that the stochastic generator $\tilde{\mathcal{G}}$ is consistent, we note that the neural network $\mathcal F_\theta$ takes as an input the pair $(x_k-x_k^{\ell+1},k/N)$ and outputs the deviation from the rollout control $u_k^{\ell+1}$. Upon completion of training, the functions $\mu_k^{\ell+1}$, $k=0,\dots,N-1$, are defined as\footnote{Actually, the policy $\mu_k^{\ell+1}(x_k)$ also involves a saturation function to ensure the control $\mu_k^{\ell+1}(x_k)$ is feasible, i.e., satisfying the acceleration limits. For simplicity, we do not include the saturation function explicitly in defining $\mu_k^{\ell+1}$, as it does not affect our analysis of the consistency property of the generator.}
$$
\mu_k^{\ell+1}(x_k)
= u_k^{\ell+1}
+\mathcal E_\theta\!\left(x_k-x_k^{\ell+1},k/N\right),
$$
where
$$
\mathcal E_\theta\!\left(x_k-x_k^{\ell+1},k/N\right)
=\mathcal F_\theta\!\left(x_k-x_k^{\ell+1},k/N\right)
-\mathcal F_\theta\!\left(0,k/N\right).
$$
Since $\mathcal E_\theta(0,k/N)=0$ by construction, for every value of $\theta$, we have $\mu_k^{\ell+1}(x_k^{\ell+1})=u_k^{\ell+1}$ for all $k$, so the generator $\tilde{\mathcal G}$ satisfies the consistency condition exactly, independently of the neural-network training accuracy. Further details on the structure of $\mathcal F_\theta$ and training can be found in the supplementary code.

We evaluate the method in three scenarios involving different initial conditions and obstacle layouts. Scenario 1 has 12 obstacles, scenario 2 has 16 obstacles, and scenario 3 has 24 obstacles; see Figs.~\ref{fig:env1_traj}, \ref{fig:env2_traj}, and \ref{fig:env3_traj}. In all scenarios, the green box denotes the goal region, while gray boxes and red spheres denote obstacles. The drone icon marks the initial condition, and the trajectories shown are the ones generated under policies $\pi^0$, $\pi^4$, and $\pi^{12}$. The training set \eqref{eq:train_set} contains $500$ state-control pairs at each iteration. It can be seen from Figs.~\ref{fig:env1_cost}, \ref{fig:env2_cost} and \ref{fig:env3_cost} that multiagent on-line PI generates significantly improved sequences of policies in all scenarios. Tables~\ref{tab:timing_both_envs} and \ref{tab:timing_parallel_both_envs} show the computational metrics when the individual control components are calculated without and with parallel processors, respectively. The metrics include the computation time per stage (mean$\pm$standard deviation, in milliseconds, with header `Stage'), and the time for data generation and  training of each new policy (in seconds, with header `$\mathcal F_\theta$').\footnote{All experiments were conducted on a machine equipped with an Apple M4 Max chip (14-core CPU, 36 GB unified memory).} 
It can be seen that even without parallel computation, the per-stage time is well below the `standard' $0.25$ second limit. Moreover, parallel processors significantly reduce the computational time.
\begin{figure}[ht!]
    \centering
    \begin{subfigure}[t]{\columnwidth}
        \centering
        \includegraphics[width=0.7\textwidth]{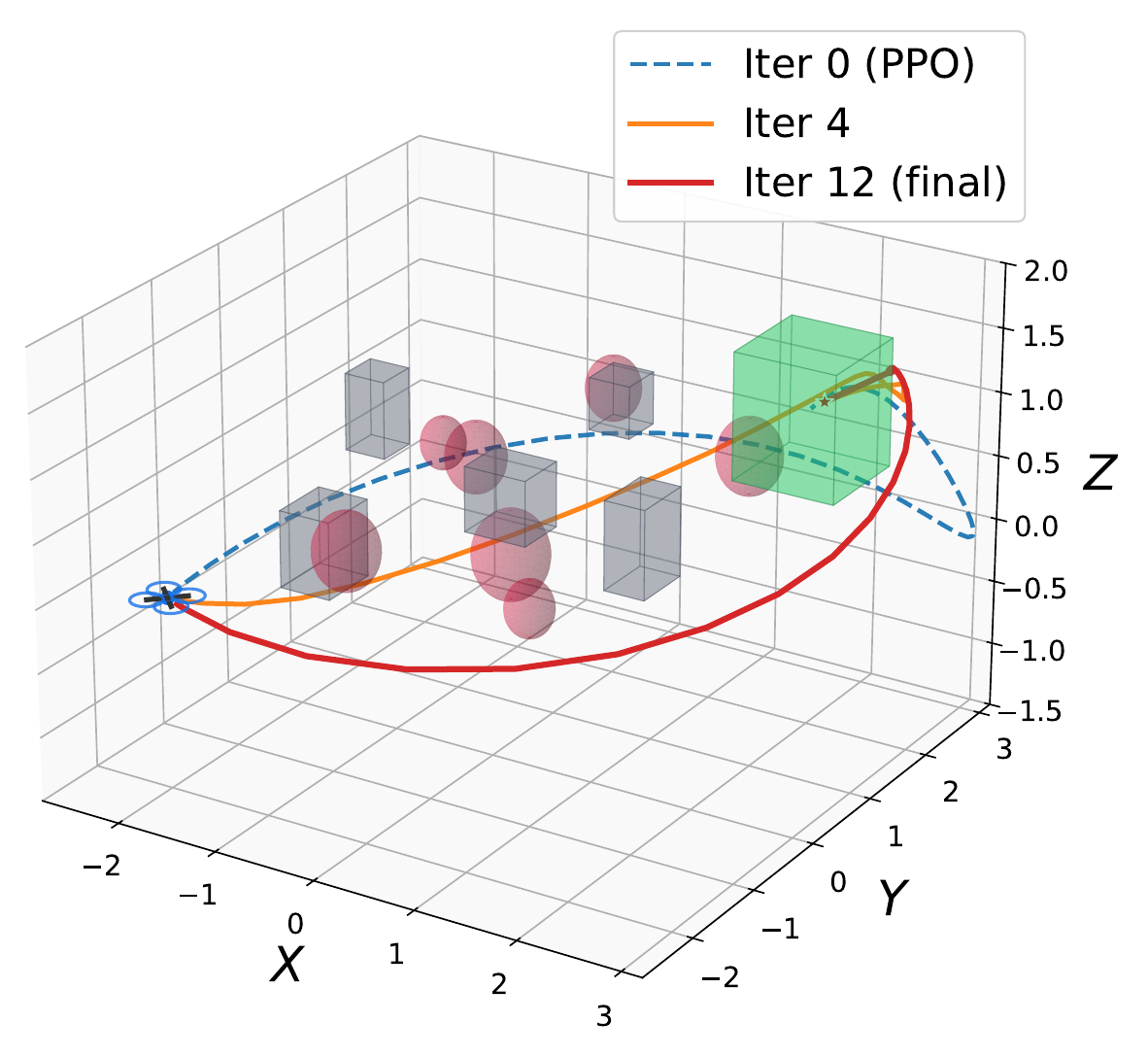}
        \caption{Scenario 1}
        \label{fig:env1_traj}
    \end{subfigure}
    \vspace{0.5cm}
    \begin{subfigure}[t]{\columnwidth}
        \centering
        \includegraphics[width=0.7\textwidth]{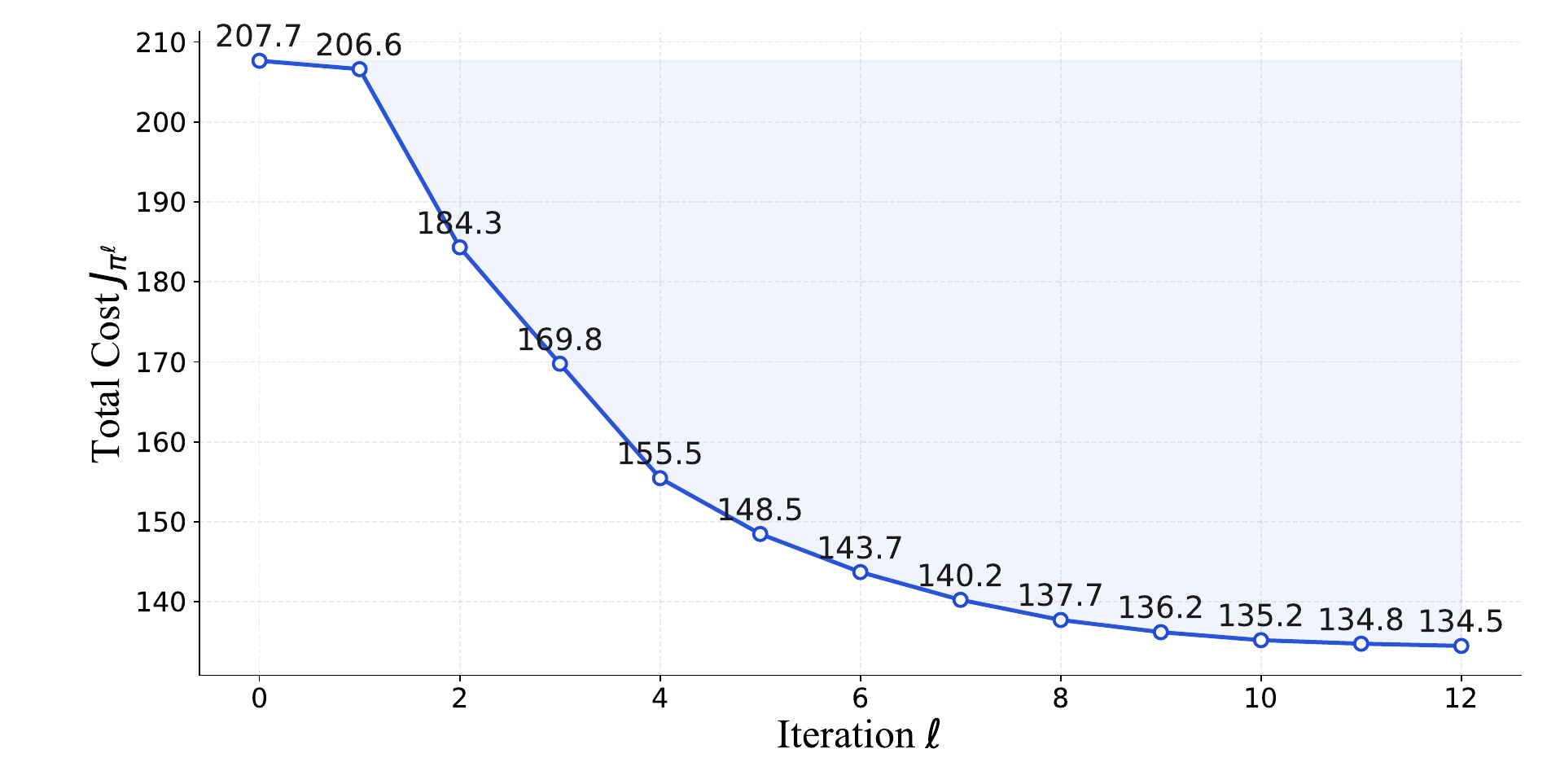}
        \caption{Multiagent on-line PI of scenario 1} 
        \label{fig:env1_cost}
    \end{subfigure}
    \caption{Multiagent on-line PI applied to scenario 1. The scenario is illustrated in figure (a), where the green region is the goal; gray boxes and red spheres are obstacles. The successive cost improvement of multiagent on-line PI is shown in figure (b).}
    \label{fig:env1_results}
\end{figure}

\begin{figure}[ht!]
    \centering
    \begin{subfigure}[t]{0.7\columnwidth}
        \centering
        \includegraphics[width=\textwidth]{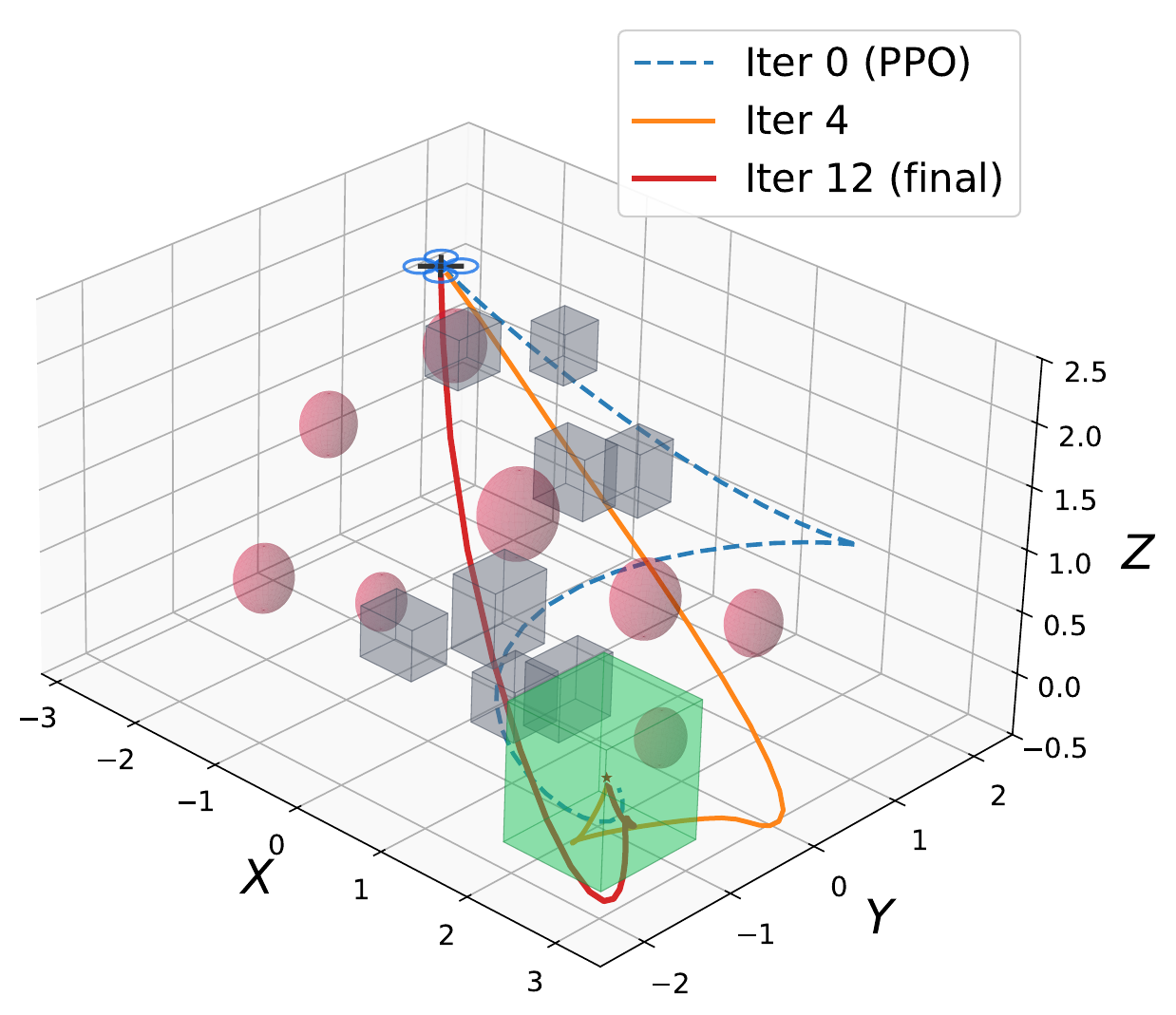}
        \caption{Scenario 2}
        \label{fig:env2_traj}
    \end{subfigure}
    \vspace{0.5cm}
    \begin{subfigure}[t]{0.7\columnwidth}
        \centering
        \includegraphics[width=\textwidth]{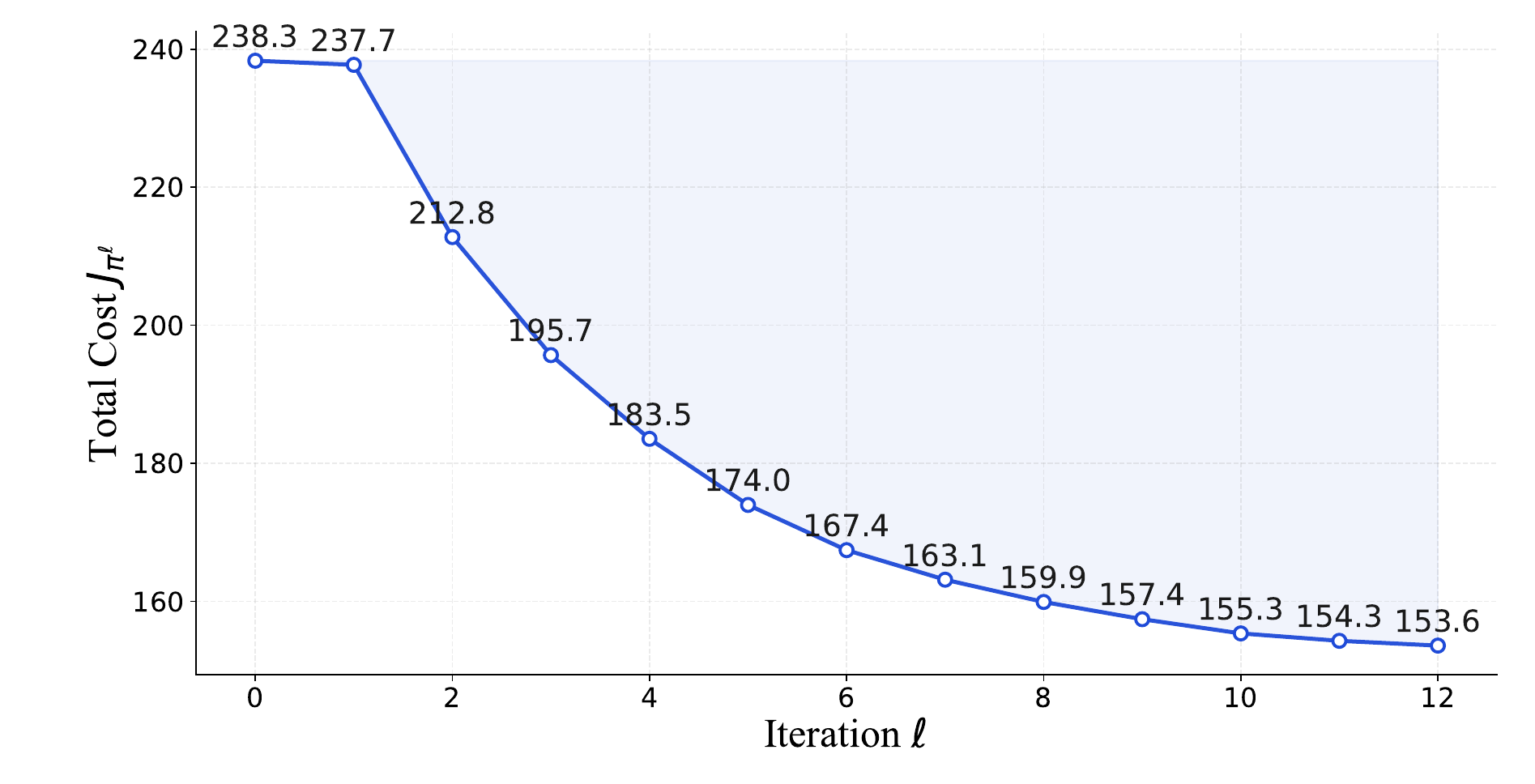}
        \caption{Multiagent on-line PI of scenario 2} 
        \label{fig:env2_cost}
    \end{subfigure}
    \caption{Multiagent on-line PI applied to scenario 2. The scenario is illustrated in figure (a), where the green region is the goal; gray boxes and red spheres are obstacles. The successive cost improvement of multiagent on-line PI is shown in figure (b).}
    \label{fig:env2_results}
\end{figure}

\begin{figure}[ht!]
    \centering
    \begin{subfigure}[t]{0.7\columnwidth}
        \centering
        \includegraphics[width=\textwidth]{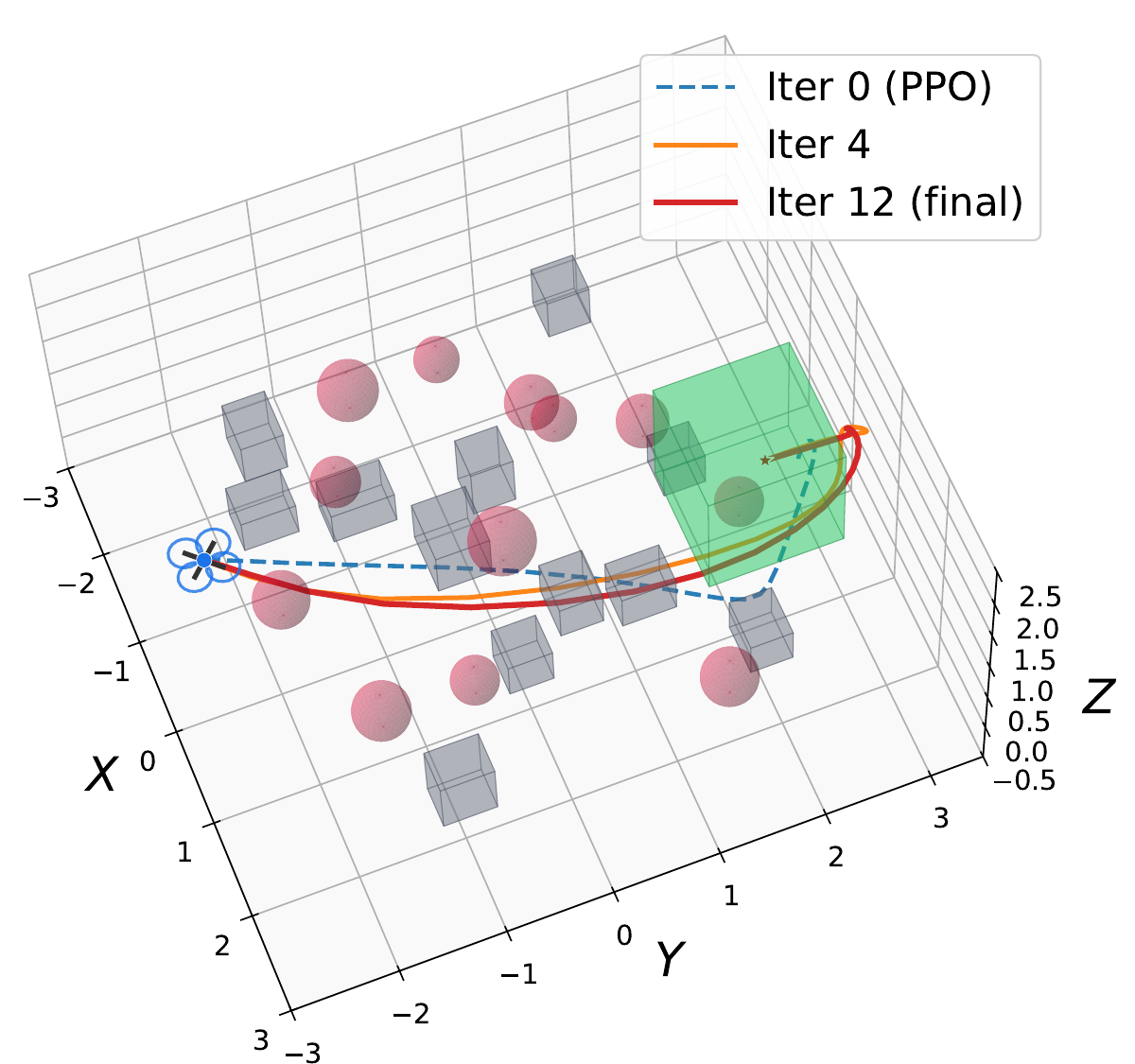}
        \caption{Scenario 3}
        \label{fig:env3_traj}
    \end{subfigure}
    \vspace{0.5cm}
    \begin{subfigure}[t]{0.7\columnwidth}
        \centering
        \includegraphics[width=\textwidth]{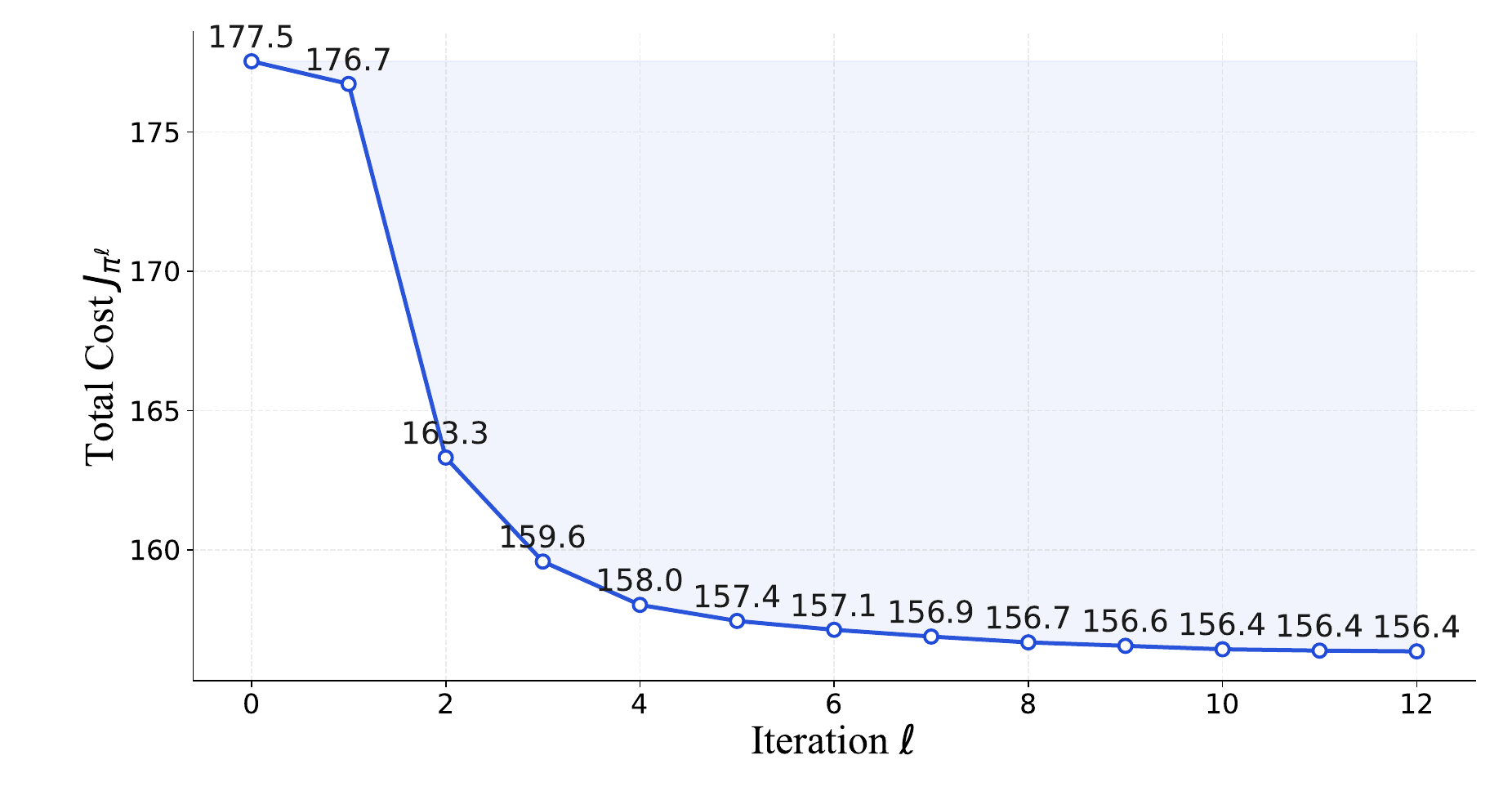}
        \caption{Multiagent on-line PI of scenario 3} 
        \label{fig:env3_cost}
    \end{subfigure}
    \caption{Multiagent on-line PI applied to scenario 3. The scenario is illustrated in figure (a), where the green region is the goal; gray boxes and red spheres are obstacles. The successive cost improvement of multiagent on-line PI is shown in figure (b).}
    \label{fig:env3_results}
\end{figure}

\begin{table}[h]
\centering
\caption{Per-iteration computational metrics without parallel processor. We report per-stage time (mean$\pm$standard deviation, in milliseconds, with header `Stage') and data collection plus training time (in seconds, with header `$\mathcal F_\theta$').}
\label{tab:timing_both_envs}
\begin{tabular}{c|cc|cc|cc}
\hline
\multirow{2}{*}{$\ell$} & \multicolumn{2}{c|}{Scenario 1} & \multicolumn{2}{c|}{Scenario 2} & \multicolumn{2}{c}{Scenario 3} \\
& Stage & $\mathcal F_\theta$ & Stage & $\mathcal F_\theta$ & Stage & $\mathcal F_\theta$ \\
\hline
1 & 18.2$\pm$9.9 & 5.1 & 24.6$\pm$13.3 & 6.8 & 34.0$\pm$18.5 & 9.0 \\
3 & 36.6$\pm$20.5 & 5.3 & 42.0$\pm$23.7 & 6.6 & 51.8$\pm$29.0 & 8.9\\
5 & 35.5$\pm$19.8 & 5.4 & 42.3$\pm$23.9 & 6.9 & 53.0$\pm$29.5 & 9.2\\
7 & 35.7$\pm$20.0 & 5.2 & 41.7$\pm$23.3 & 7.0 & 52.3$\pm$29.2 & 9.2\\
9 & 36.4$\pm$20.5 & 5.3 & 41.8$\pm$23.8 & 6.6 & 52.6$\pm$28.5 & 9.2\\
12 & 35.6$\pm$20.2 & 5.5 & 42.4$\pm$23.7 & 7.2 & 52.3$\pm$29.3 & 9.4\\
\hline
\end{tabular}
\end{table}

\begin{table}[h]
\centering
\caption{Per-iteration computational metrics with parallel processors. We report per-stage time (mean$\pm$standard deviation, in milliseconds, with header `Stage') and data collection plus training time (in seconds, with header `$\mathcal F_\theta$').}
\label{tab:timing_parallel_both_envs}
\begin{tabular}{c|cc|cc|cc}
\hline
\multirow{2}{*}{$\ell$} & \multicolumn{2}{c|}{Scenario 1} & \multicolumn{2}{c|}{Scenario 2} & \multicolumn{2}{c}{Scenario 3} \\
& Stage & $\mathcal F_\theta$ & Stage & $\mathcal F_\theta$ & Stage & $\mathcal F_\theta$ \\
\hline
1 & 5.0$\pm$2.2 & 5.1 & 6.2$\pm$2.9 & 6.9 & 9.8$\pm$4.6 & 9.1 \\
3 & 8.8$\pm$5.2 & 5.7 & 9.7$\pm$5.1 & 6.7 & 11.9$\pm$6.7 & 9.0\\
5 & 8.6$\pm$4.9 & 5.5 & 10.3$\pm$6.9 & 6.9 & 11.6$\pm$6.3 & 9.1\\
7 & 9.3$\pm$6.7 & 5.3 & 10.3$\pm$6.9 & 7.0 & 11.5$\pm$6.2 & 9.1\\
9 & 8.5$\pm$4.3 & 5.2 & 9.8$\pm$5.1 & 6.6 & 11.6$\pm$6.1 & 8.9\\
12 & 8.9$\pm$5.7 & 5.5 & 9.9$\pm$5.4 & 7.1 & 11.3$\pm$5.8 & 9.2\\
\hline
\end{tabular}
\end{table}

\subsection{Path Planning and Coordination of Multiple Drones}\label{sec:4_3}
We consider an extension of the problem described in Section~\ref{sec:4_2}, where we aim to compute paths for $m$ drones to reach their respective goal regions while avoiding collisions among themselves. In particular, each drone is characterized by its position and acceleration in $3$ dimensions as described in Section~\ref{sec:4_2}. Therefore, the state $x_k$ in the present problem is a $6m$-dimensional vector. Similarly, the control $u_k$ is a $3m$-dimensional vector, representing the acceleration of all drones in three directions. The state equation $f_k$ and control constraint set $U_k(x_k)$ can be defined accordingly. The stage cost $g_k(x_k,u_k)$ is the sum of costs contributed by the individual drones, as defined in Section~\ref{sec:4_2}, plus additional costs that penalize short pairwise distances between drones. We assume that all drones begin operation at the same time. This assumption is made without loss of generality: if some drones are scheduled to start moving later than others, they can be modeled as staying still at their initial position until their respective start times.

We have applied the stochastic version of multiagent on-line PI to solve this problem, where there are $3m$ `agents' in our terminology. The way in which the stochastic generator $\tilde{\mathcal{G}}$ is defined is identical to that described in Section~\ref{sec:4_2}, except that suitable modifications were made to accommodate the larger problem dimension. In particular, the training set \eqref{eq:train_set} contains $600$, $700$, and $800$ state-control pairs at each iteration for the two-, three-, and four-drone problems, respectively. Although we have modeled the problem by using a single state $x_k$ and a single control $u_k$, we note that our algorithm can be carried out in a distributed fashion. In particular, at state $x_k^{\ell+1}$, the first drone can compute the control components $(u_{k,1},u_{k,2},u_{k,3})$ related to itself, provided that it has access to the state information $x_k^{\ell+1}$ and the policy $\pi^\ell$. Upon computing the control components $(u_{k,1}^{\ell+1},u_{k,2}^{\ell+1},u_{k,3}^{\ell+1})$, the first drone can share these components with all the other drones, so that the second drone can compute its own control components; see Fig.~\ref{fig:multi_drone}. As indicated in Section~\ref{sec:4_2}, each drone can use parallel computation when calculating individual control components. However, the drones cannot compute the control components associated with themselves in parallel. Instead, a drone needs to wait for the computation results of the preceding drones, according to some predefined order.

\begin{figure}[t]
    \centering
    \includegraphics[width=0.9\linewidth]{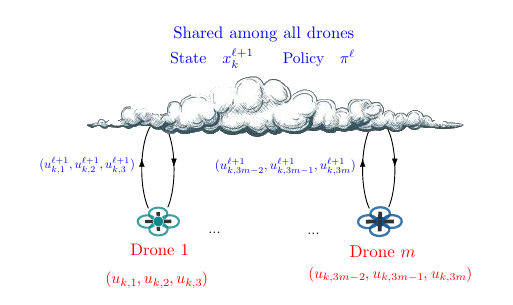}
    \caption{Illustration of multiagent on-line PI via distributed implementation. During iteration $(\ell+1)$, at state $x_k^{\ell+1}$, all $m$ drones have access to the state $x_k^{\ell+1}$ (including positions and velocities of other drones), and to the policy $\pi^\ell$ computed at the $\ell$th iteration (so they can simulate trajectories of all drones), stored in the cloud. The first drone computes control components $(u_{k,1},u_{k,2},u_{k,3})$ related to itself, one component at a time. The computed control components $(u_{k,1}^{\ell+1},u_{k,2}^{\ell+1},u_{k,3}^{\ell+1})$ are uploaded to the cloud, and the second drone can compute the control components $(u_{k,4},u_{k,5},u_{k,6})$, after receiving the control components computed by the first drone. This process is continued until the $m$th drone computes its control.}
    \label{fig:multi_drone}
\end{figure}

We have considered the same obstacle layouts as for the scenarios 1, 2, and 3, introduced in Section~\ref{sec:4_2}, and we have computed corresponding paths for 2, 3, and 4 drones, respectively. The obstacle layout, the position of the goal regions of all the drones, the trajectories generated under the initial policy, and the policies obtained during the $4$th and the final iterations are shown in Figs.~\ref{fig:env1_mas_traj}, \ref{fig:env2_mas_traj}, and \ref{fig:env3_mas_traj}. It can be seen from Figs.~\ref{fig:env1_mas_cost}, \ref{fig:env2_mas_cost} and \ref{fig:env3_mas_cost} that multiagent on-line PI generates significantly improved sequences of policies in  all the test cases. The corresponding computation times are given in Table~\ref{tab:timing_mas_both_envs}, where parallel computation is used for calculating individual control components. The computation time per stage (mean$\pm$standard deviation, in milliseconds) is below the `standard' $0.25$ second limit for all problems, even for the case of $4$ drones. The time for data generation and  training of each new policy is also shown in Table~\ref{tab:timing_mas_both_envs} (in seconds).

\begin{figure}[ht!]
    \centering
    \begin{subfigure}[t]{\columnwidth}
        \centering
        \includegraphics[width=0.8\textwidth]{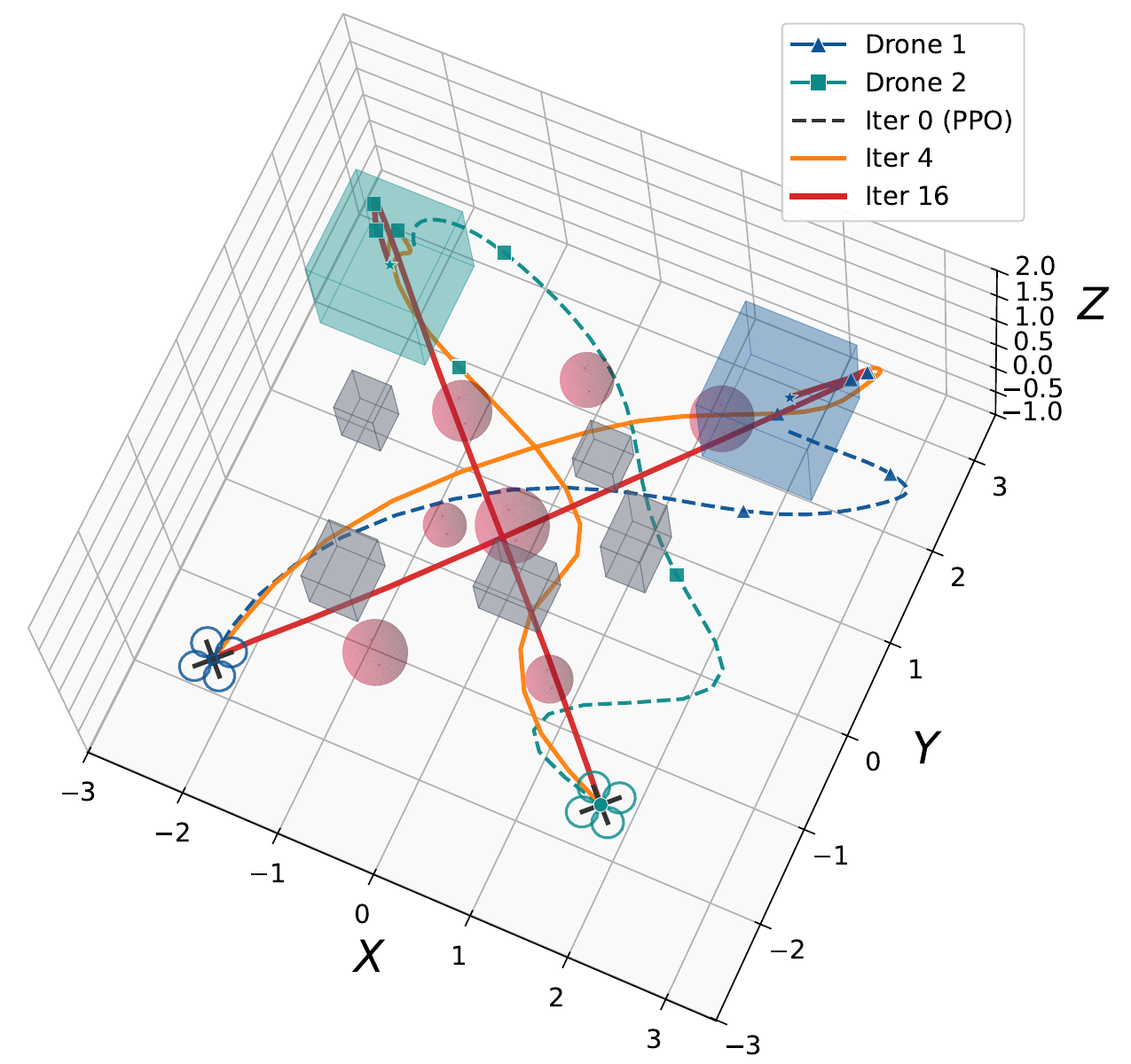}
        \caption{Scenario 1 with 2 drones}
        \label{fig:env1_mas_traj}
    \end{subfigure}
    \vspace{0.5cm}
    \begin{subfigure}[t]{\columnwidth}
        \centering
        \includegraphics[width=0.8\textwidth]{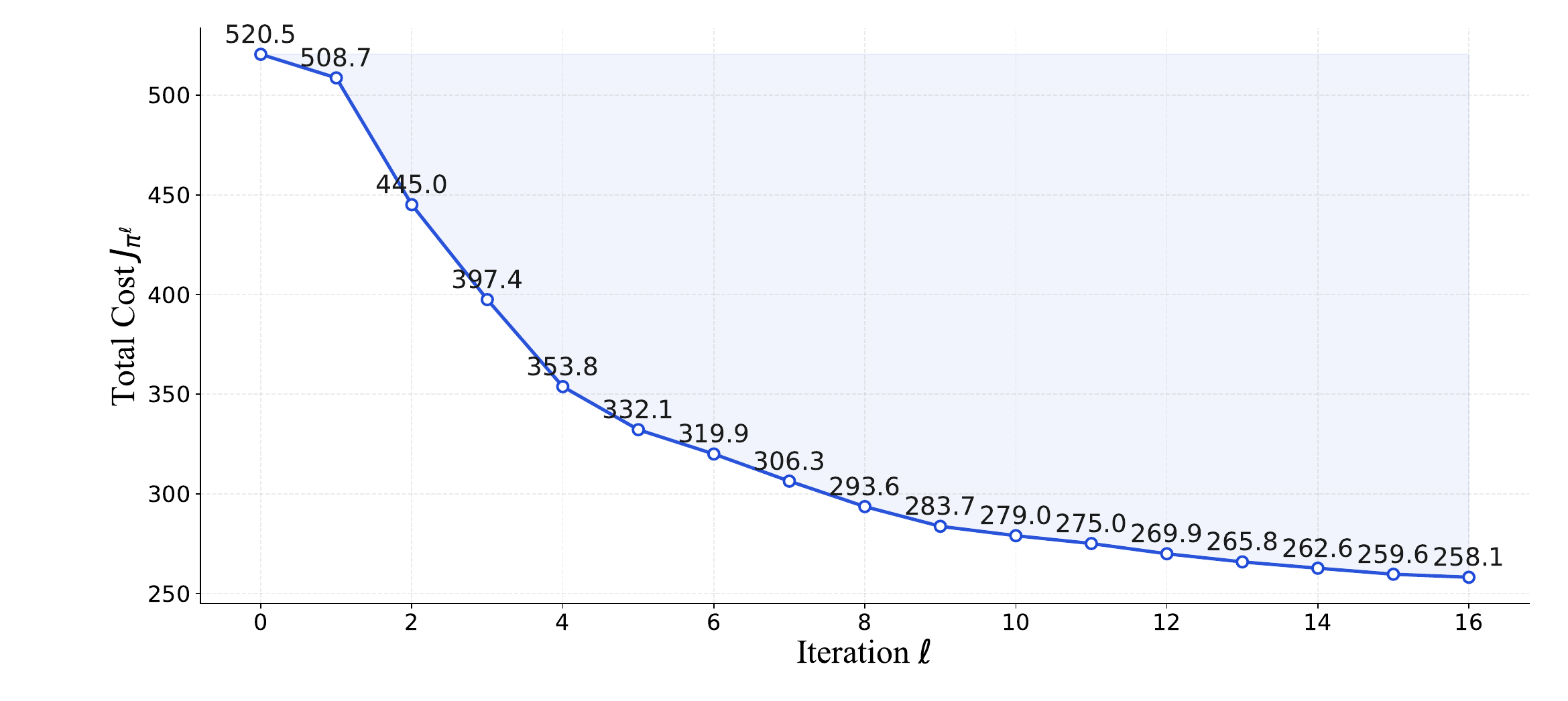}
        \caption{Multiagent on-line PI of scenario 1} 
        \label{fig:env1_mas_cost}
    \end{subfigure}
    \caption{Multiagent on-line PI applied to scenario 1. The scenario is illustrated in figure (a), where the green and blue regions are the goals for the two drones; gray boxes and red spheres are obstacles. The successive cost improvement of multiagent on-line PI is shown in figure (b).}
    \label{fig:env1_mas_results}
\end{figure}

\begin{figure}[ht!]
    \centering
    \begin{subfigure}[t]{\columnwidth}
        \centering
        \includegraphics[width=0.8\textwidth]{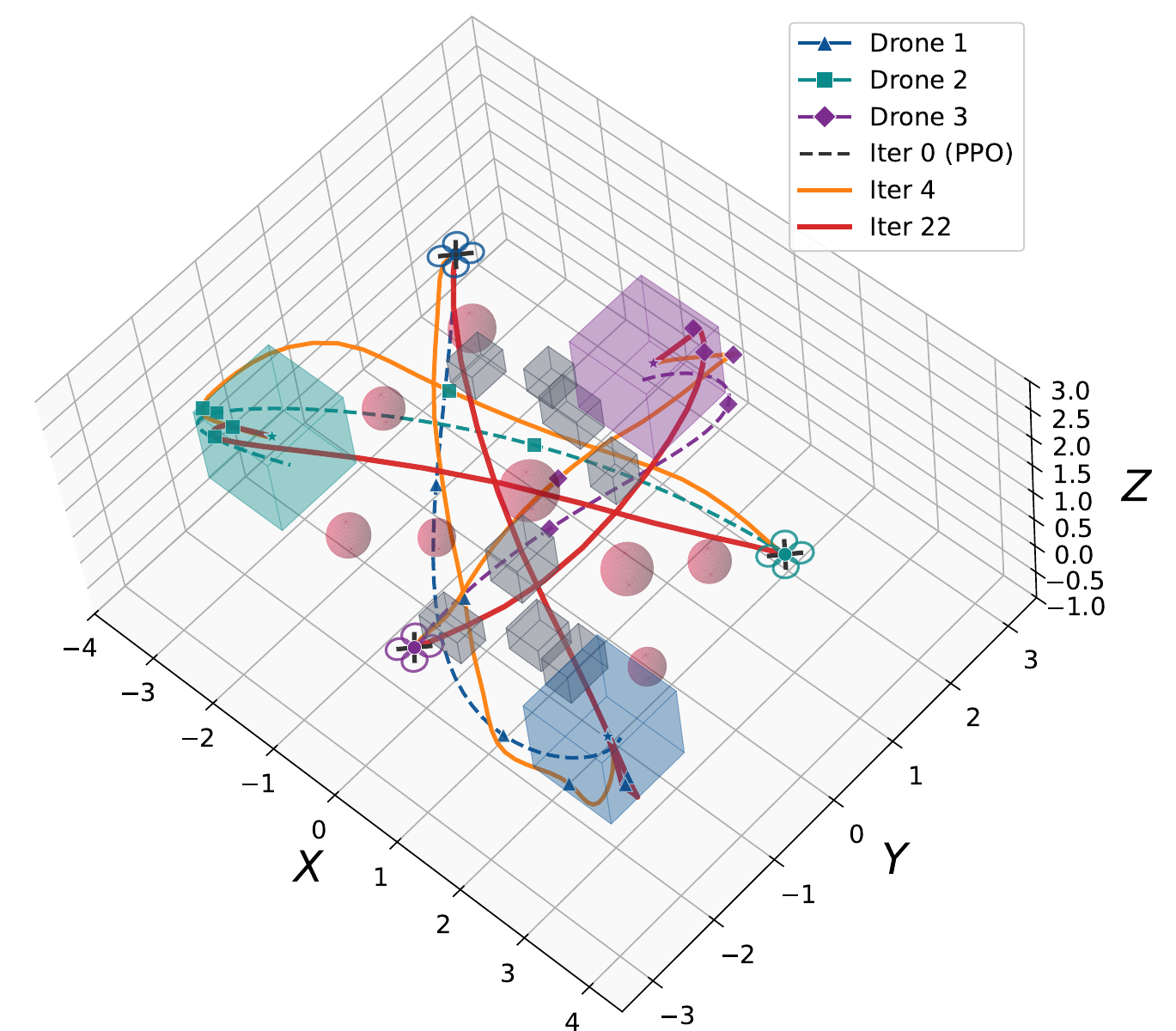}
        \caption{Scenario 2 with 3 drones}
        \label{fig:env2_mas_traj}
    \end{subfigure}
    \vspace{0.5cm}
    \begin{subfigure}[t]{\columnwidth}
        \centering
        \includegraphics[width=0.8\textwidth]{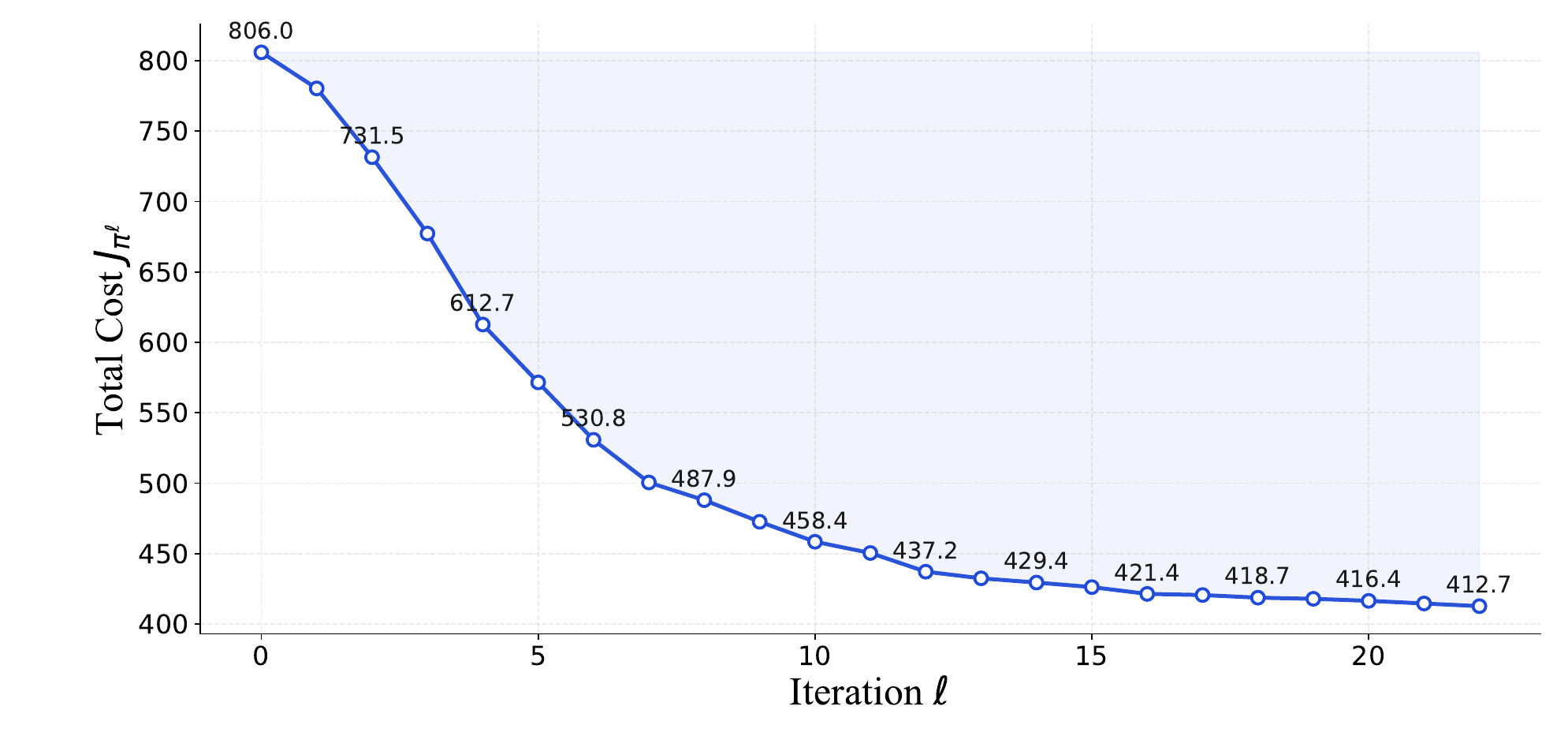}
        \caption{Multiagent on-line PI of scenario 2} 
        \label{fig:env2_mas_cost}
    \end{subfigure}
    \caption{Multiagent on-line PI applied to scenario 2. The scenario is illustrated in figure (a), where the green, blue, and purple regions are the goals for the three drones; gray boxes and red spheres are obstacles. The successive cost improvement of multiagent on-line PI is shown in figure (b).}
    \label{fig:env2_mas_results}
\end{figure}

\begin{figure}[ht!]
    \centering
    \begin{subfigure}[t]{\columnwidth}
        \centering
        \includegraphics[width=0.8\textwidth]{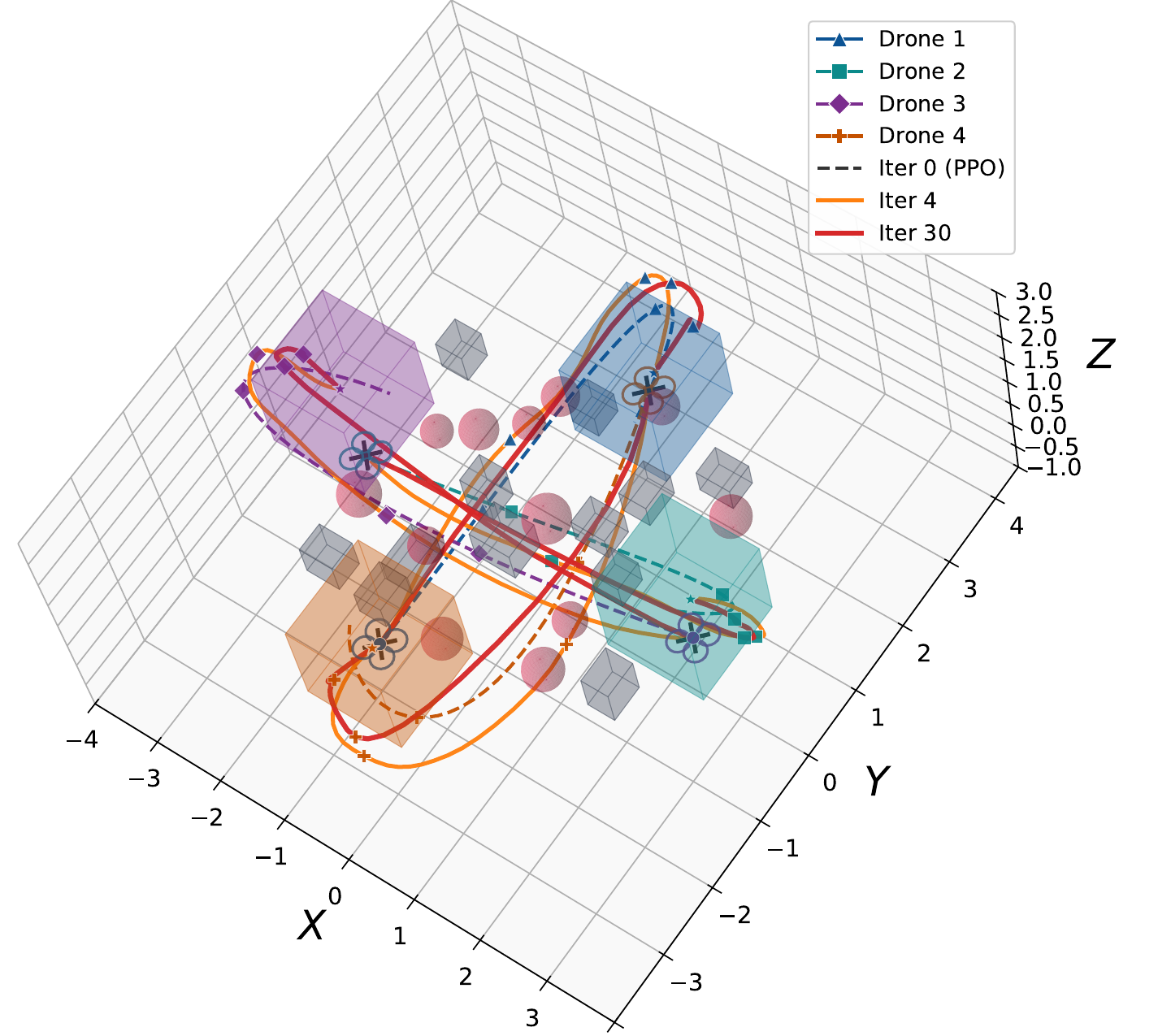}
        \caption{Scenario 3 with 4 drones}
        \label{fig:env3_mas_traj}
    \end{subfigure}
    \vspace{0.5cm}
    \begin{subfigure}[t]{\columnwidth}
        \centering
        \includegraphics[width=0.8\textwidth]{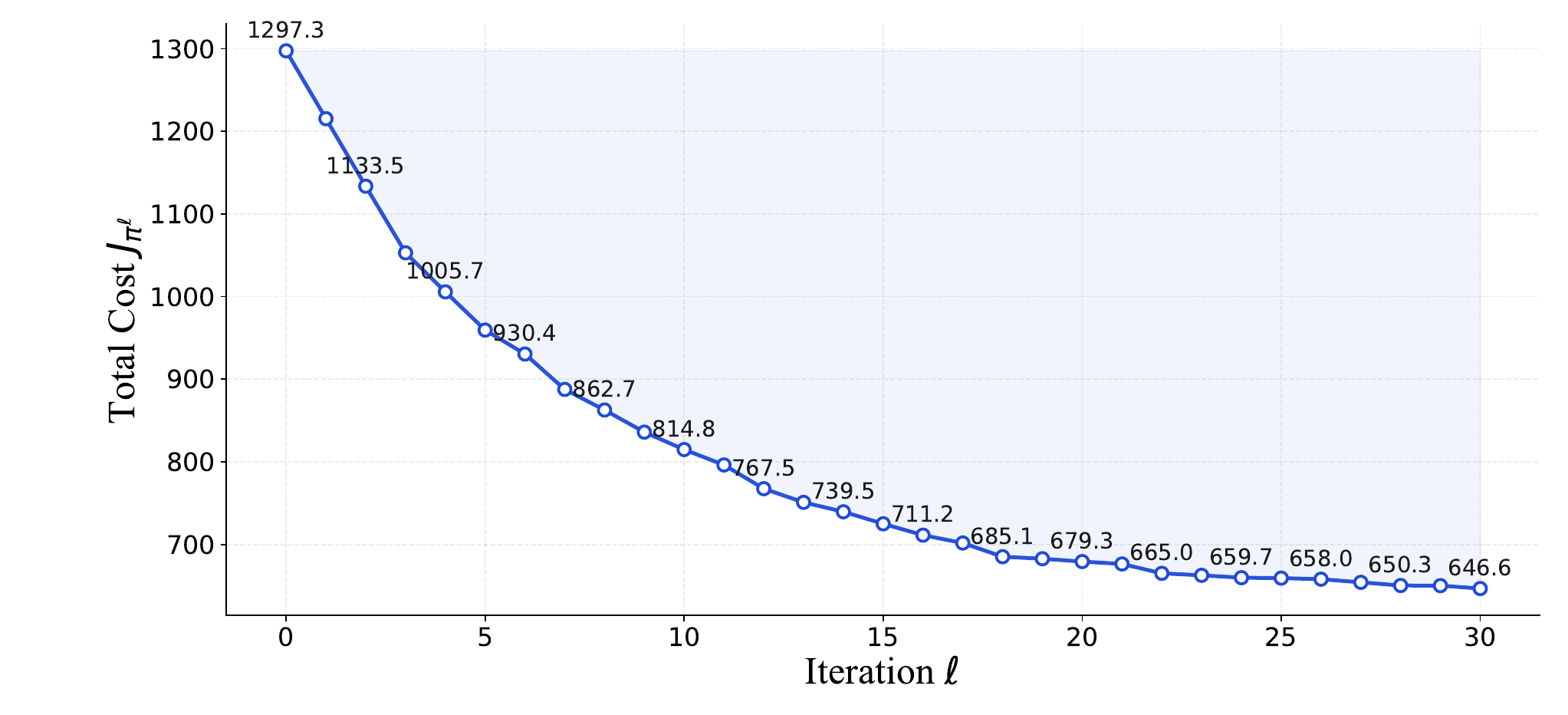}
        \caption{Multiagent on-line PI of scenario 3} 
        \label{fig:env3_mas_cost}
    \end{subfigure}
    \caption{Multiagent on-line PI applied to scenario 3. The scenario is illustrated in figure (a), where the green, blue, purple  and orange regions are the goals for the four drones; gray boxes and red spheres are obstacles. The successive cost improvement of multiagent on-line PI is shown in figure (b).}
    \label{fig:env3_mas_results}
\end{figure}

\begin{table}[ht]
\centering
\caption{Per-iteration computational metrics with parallel processors. We report per-stage time (mean$\pm$standard deviation, in milliseconds, with header `Stage') and data collection plus training time (in seconds, with header `$\mathcal F_\theta$').}
\label{tab:timing_mas_both_envs}
\begin{tabular}{c|cc|cc|cc}
\hline
\multirow{2}{*}{$\ell$} & \multicolumn{2}{c|}{2 Drones} & \multicolumn{2}{c|}{3 Drones} & \multicolumn{2}{c}{4 Drones} \\
& Stage & $\mathcal F_\theta$ & Stage & $\mathcal F_\theta$ & Stage & $\mathcal F_\theta$ \\
\hline
1 & 15.1±7.1 & 16.6 & 40.6$\pm$20.4 & 37.9 & 92.8$\pm$47.5 & 80.0\\
4 & 24.2±12.3 & 16.8 & 53.3$\pm$27.8 & 39.5 & 109.0$\pm$57.4 & 80.1\\
7 & 24.1±12.3 & 16.6 & 53.1$\pm$27.7 & 37.5 & 109.2$\pm$57.4 & 82.1\\
10 & 24.1±12.3 & 16.6 & 52.9$\pm$27.4 & 38.7 & 110.3$\pm$57.7 & 74.8 \\
13 & 24.8±12.5 & 17.1 & 53.3$\pm$27.7 & 39.0 & 109.5$\pm$57.5 & 81.0\\
16 & 24.7±12.8 & 17.0 & 53.2$\pm$27.7 & 39.3 & 111.6$\pm$59.7 & 78.4\\
19 & - & - & 53.2$\pm$28.2 & 39.4 & 109.3$\pm$57.4 & 83.1\\
22 & - & - & 53.7$\pm$27.6 & 37.2 & 109.9$\pm$57.7 & 76.8\\
25 & - & - & - & - & 109.9$\pm$57.7 & 80.9\\
28 & - & - & - & - & 109.2$\pm$57.4 & 79.1\\
30 & - & - & - & - & 109.7$\pm$57.5 & 79.6\\
\hline
\end{tabular}
\end{table}

\subsection{Path Planning for a Robotic Arm}\label{sec:4_4}
Let us consider a path planning problem for a lightweight robotic arm, known as the Franka Emika Robot. It has $7$ degrees of freedom for moving the end of the arm (known as the \emph{end effector} and abbreviated as EE). The state $x_k$ is a $14$-dimensional vector, representing the positions and velocities of $7$ joints. The control $u_k$ is composed of the accelerations at each joint. For the state equation $f_k$, we use the simulator provided by Robotics Toolbox for Python \cite{rtb}. The analytical form of $f_k$ is not necessary for our purpose, as we rely on simulated trajectories in our algorithm. The stage costs $g_k$ are defined by using $28$ interpolation points between the joints and the EE point. These points, together with the EE point, are used to specify a collision penalty and the distance from the EE to the goal region. Further details are given in the supplementary code.    

Similar to the preceding two sections, we have applied the stochastic version of multiagent on-line PI to solve the problem. In particular, we have used the same approach to define the stochastic generator $\tilde{\mathcal{G}}$, except that the size of the data set and the form of the neural network are modified to accommodate changes in state and control dimensions. The training set \eqref{eq:train_set} contains $400$ state-control pairs at each iteration. The trajectories of the EE under policies $\pi^0$, $\pi^4$, and $\pi^8$, starting with an initial policy obtained via PPO,  are shown in Fig.~\ref{fig:arm_traj}, where the robotic arm is illustrated by the gray cylinders, the end effector is shown as the red cylinder, the obstacles are shown as red spheres, and the goal region is the green rectangle. It can be seen that EE under $\pi^0$ does not reach the goal region at stage $N$, while it reaches the goal region under both $\pi^4$ and $\pi^8$. The successive cost improvement of multiagent on-line PI is shown in Fig.~\ref{fig:arm_cost}. For implementation convenience, only a serial processor is used in our computational studies here. The computation time per stage is around $55\pm31$ (mean$\pm$standard deviation, in milliseconds), which is close to satisfying the $50$ millisecond limit. Had parallel processors been used, the computation time would be well below the limit, as indicated by the computational studies in Section~\ref{sec:4_2}.

\begin{figure}[ht!]
    \centering
    \begin{subfigure}[t]{\columnwidth}
        \centering
        \includegraphics[width=0.7\textwidth]{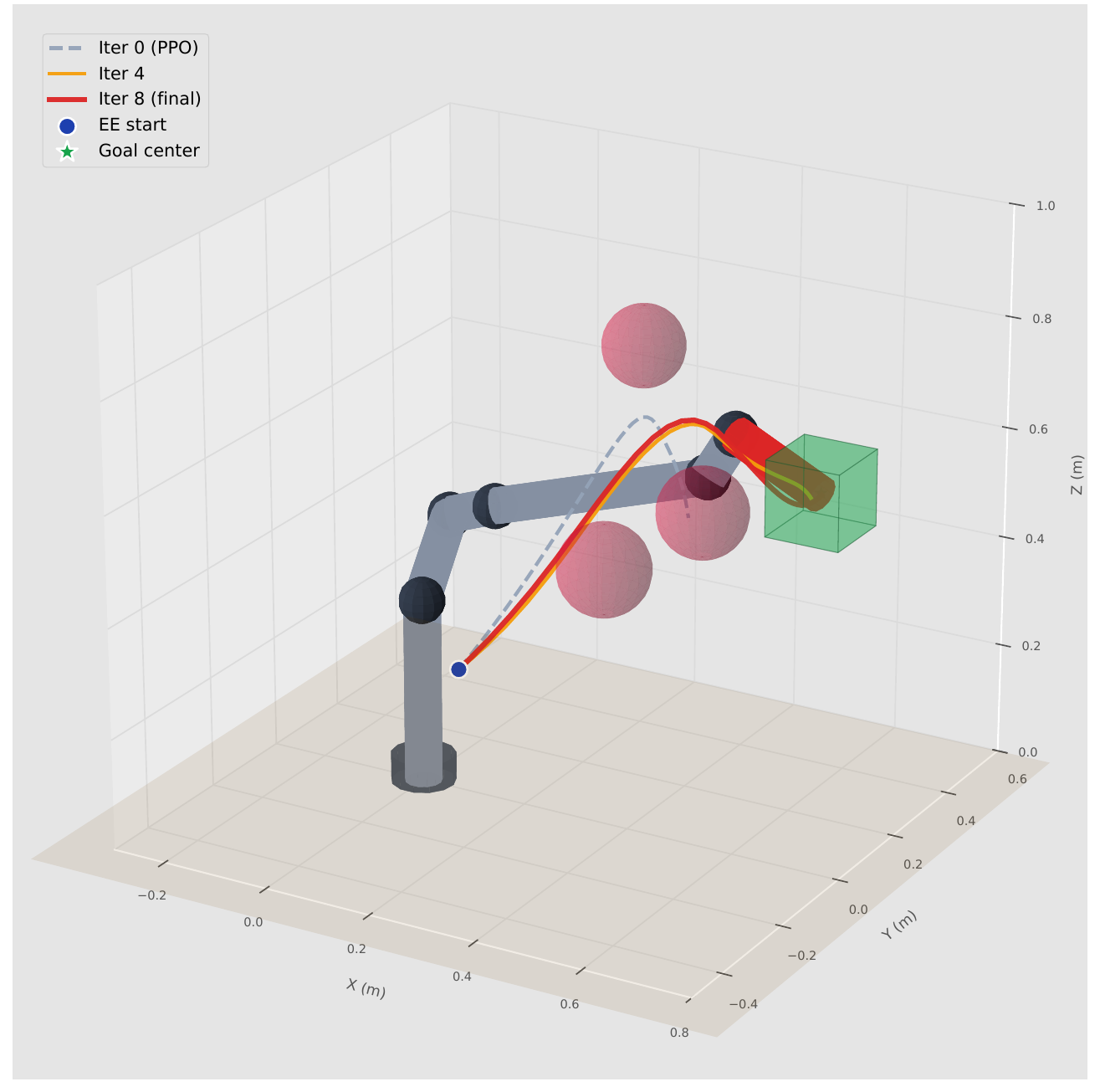}
        \caption{Path planning for robotic arm}
        \label{fig:arm_traj}
    \end{subfigure}
    \vspace{0.5cm}
    \begin{subfigure}[t]{\columnwidth}
        \centering
        \includegraphics[width=0.7\textwidth]{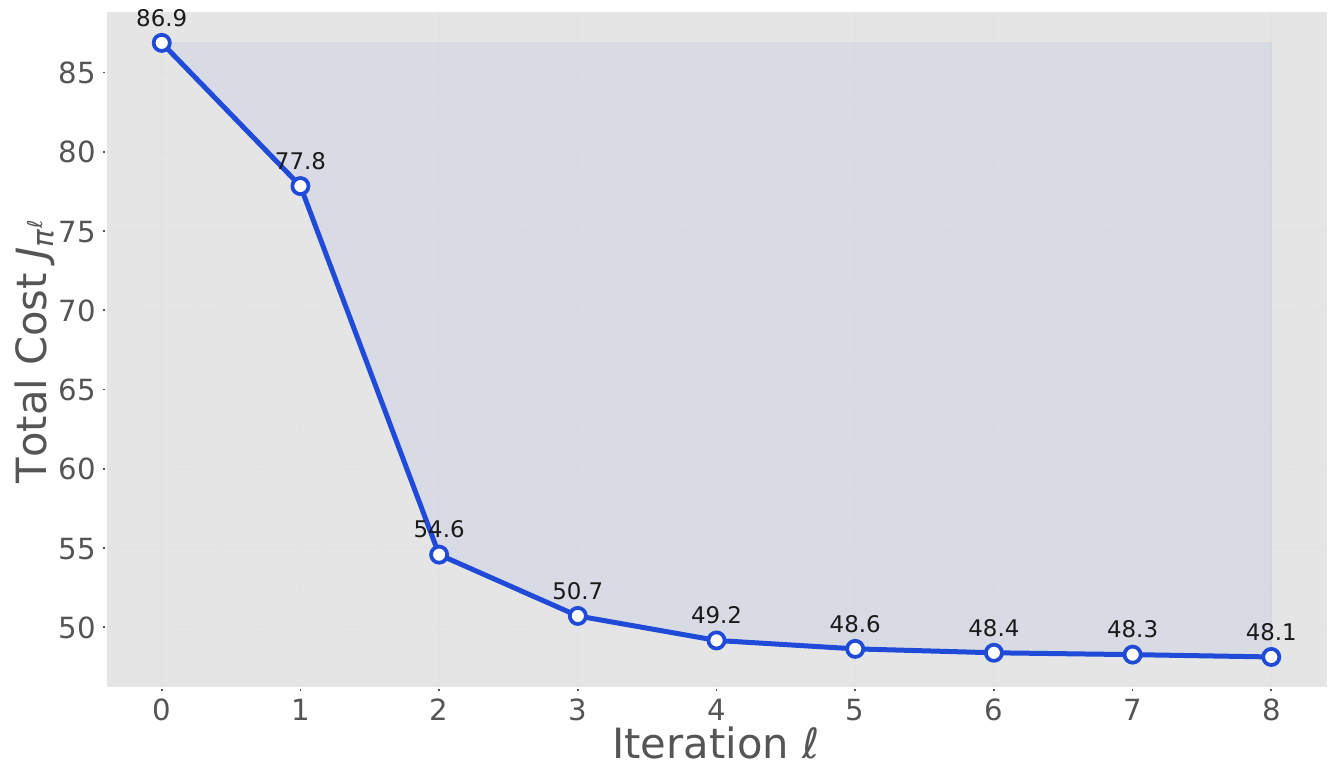}
        \caption{Multiagent on-line PI cost over iterations} 
        \label{fig:arm_cost}
    \end{subfigure}
    \caption{Multiagent on-line PI applied to path planning for a robotic arm. The scenario is illustrated in figure (a), where the green region is the goal, and red spheres are obstacles. The robotic arm is illustrated as gray cylinders and the end effector is shown as a red cylinder. The successive cost improvement of multiagent on-line PI is shown in figure (b).}
    \label{fig:arm_results}
\end{figure}

\subsection{Comparison with Learning Model Predictive Control}
Let us now compare our on-line PI method and the learning model predictive control (LMPC) method introduced in \cite{rosolia2017learning,rosolia2019learning,rosolia2022optimality}. The methods bear some similarity, since both apply to repetitive tasks with a fixed initial state, and aim at iterative policy improvement.

An important conceptual difference of on-line PI from LMPC is that it transfers improvement from one iteration to the next by trajectory-driven policy generation. In particular, we use rollout to compute an improved trajectory, generate a local state-control data set around that trajectory, and train the next policy from this data. At the same time, our consistency condition guarantees that the learned policy inherits the trajectory improvement at the fixed initial state. By contrast, LMPC uses stored successful trajectories to build a terminal safe set and terminal cost approximation inside a receding-horizon optimization scheme. 

We will first provide a computational comparison of on-line PI and  LMPC using the path-planning problems of Sections~\ref{sec:4_2}-\ref{sec:4_4}. For the single- and multi-drone path-planning problems and the robotic arm planning problem studied in Sections~\ref{sec:4_2}-\ref{sec:4_4}, on-line PI appears to be better suited. In particular, we have considered two types of initial policies. The first consists of the `good' policies used in Sections~\ref{sec:4_2}-\ref{sec:4_4}, obtained via PPO. The second consists of `poor' policies, also obtained via PPO but with substantially fewer training iterations. In all these examples, on-line PI yields larger cost reduction than LMPC over the same number of iterations, and the advantage becomes more pronounced when the initial policy is poor. As an illustration, consider path planning for 3 drones in scenario 2. Starting from a policy $\pi^0$ with cost $18493.9$, the cost of $\pi^2$ obtained via on-line PI is $914.8$, whereas the cost of the policy obtained via LMPC after 2 iterations is $6835.5$. The corresponding trajectories are shown in Fig.~\ref{fig:3_drone_compare}. Similar behavior was observed in all test cases for single- and multi-drone path planning and robotic arm planning.

We have also tested on-line PI and  LMPC on the constrained linear quadratic problem and the minimum-time problem of \cite[Secs.~IV-A-B]{rosolia2017learning}. When both methods are initialized from the same feasible trajectory, LMPC achieves larger cost reduction over the same number of iterations. This is consistent with the fact that LMPC is particularly effective for problems that are well suited for the construction of a safe set constraint and for an efficient multistep MPC optimization. Thus it appears that because the two methods aim to improve performance through different mechanisms, they may be best suited for different problem types.

\begin{figure}[t]
    \centering
    \includegraphics[width=0.8\linewidth]{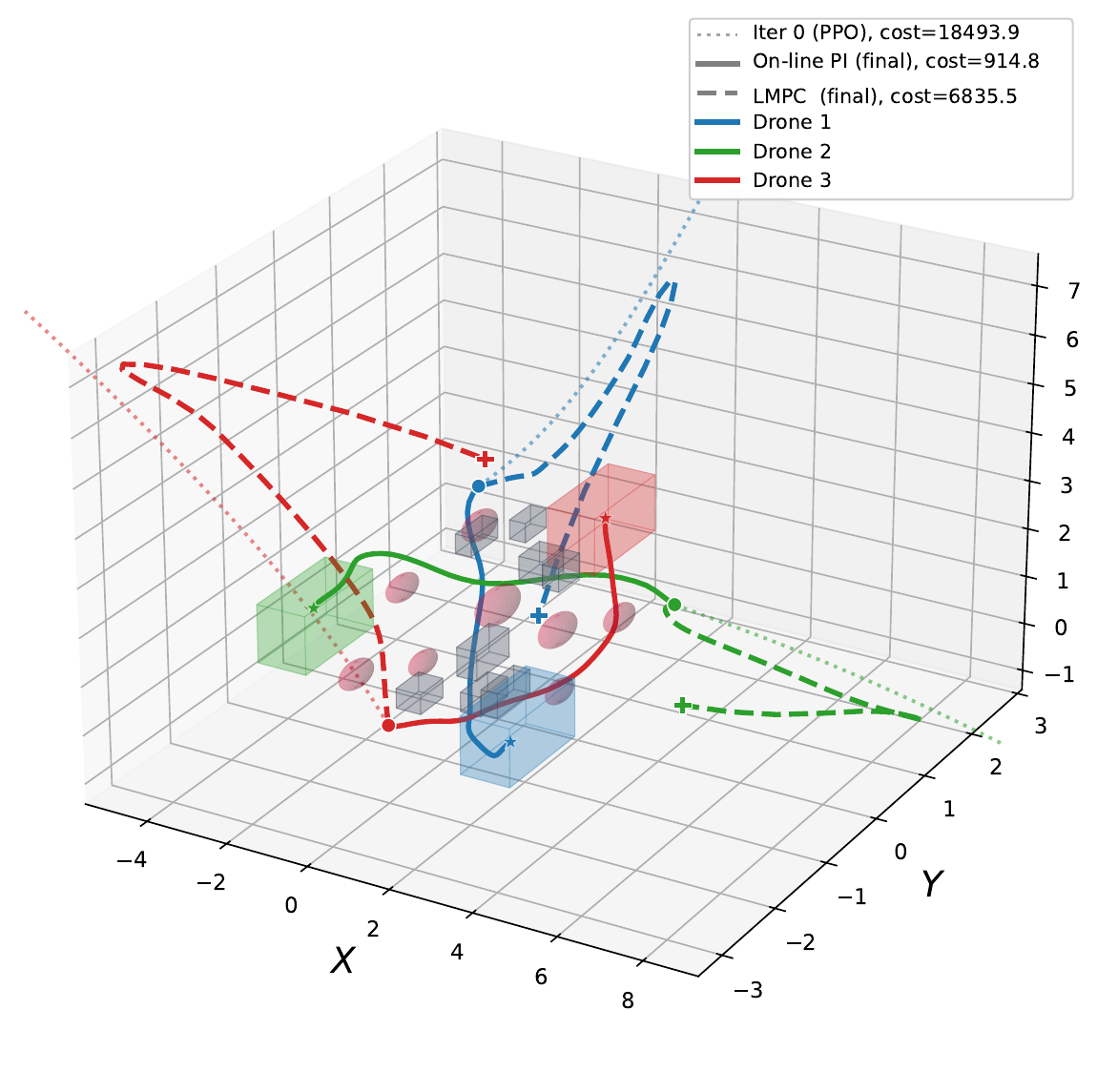}
    \caption{Illustration of trajectories of $3$ drones under the `poor' initial policy $\pi^0$, the policy $\pi^2$ computed via on-line PI, and the policy computed via LMPC after two iterations. It can be seen that on-line PI is more robust to the poor initial policy and has a larger cost improvement than LMPC.}
    \label{fig:3_drone_compare}
\end{figure}

\section{Concluding Remarks}
We have considered finite-horizon deterministic optimal control problems with a fixed initial state and proposed an on-line PI method based on trajectory-driven policy generation. The method produces a sequence of policies and corresponding trajectories. Each trajectory is used in turn to generate data for training a new policy. Under a natural consistency condition on the policy generator, we showed that the policies obtained have monotonically decreasing cost at the initial state. 
We also introduced a simplified variant that reduces real-time computation while preserving the cost improvement property, and we provided stochastic extensions of our algorithms.  In our robotic path planning experiments, we used the PPO algorithm to generate the initial base policy, but any other suitable algorithm can be used for this purpose. 

Our framework combines rollout and policy iteration ideas with flexible trajectory-based policy representations and supports neural-network-based policies. Computational results demonstrate its effectiveness in both discrete and continuous problems.  An interesting subject for further research and numerical experimentation is the robustness of our framework in the absence of an exact system model.

\bibliographystyle{alpha}
\bibliography{ref} 
\end{document}